\documentclass[11pt]{article}
\usepackage[margin=1in]{geometry}
\usepackage{authblk}
\usepackage[english]{babel}
\usepackage{amsmath,amssymb,amsthm}
\usepackage{array}
\usepackage[caption=false,font=normalsize,labelfont=sf,textfont=sf]{subfig}
\usepackage{textcomp}
\usepackage{stfloats}
\usepackage{url}
\usepackage{verbatim}
\usepackage{graphicx}
\usepackage{cite}
\usepackage{xcolor}
\usepackage[colorlinks=true, linkcolor=blue, urlcolor=blue, citecolor=blue]{hyperref}
\usepackage{float}
\usepackage{algorithm}
\usepackage{algpseudocode}

\newtheorem{theorem}{Theorem}
\newtheorem{prop}{Proposition}
\newtheorem{aspt}{Assumption}
\newtheorem{corollary}{Corollary}
\newtheorem{rem}{Remark}
\newtheorem{lem}{Lemma}
\newtheorem{defn}{Definition}

\newcommand{\reals}{\mathbb{R}}
\newcommand{\E}{\mathbb{E}}
\newcommand{\bbP}{\mathbb{P}}
\newcommand{\zero}{\boldsymbol{0}}
\newcommand{\one}{\boldsymbol{1}}
\newcommand{\Ghat}{\hat{G}}
\newcommand{\Bhat}{\hat{B}}
\newcommand{\eG}{\text{error}_G}
\newcommand{\eB}{\text{error}_B}
\newcommand{\eW}{\text{error}_W}
\newcommand{\td}{\text{d}}
\newcommand{\vhat}{\hat{v}}
\newcommand{\lambdahat}{\hat{\lambda}}

\begin{document}

\title{Inlier Recovery for Robust Registration via Gram-Matrix Overlap}

\author{
Ruizi~Wu,
Yuehaw~Khoo,
and Wanjie~Wang%
\thanks{Ruizi Wu is with the Department of Statistics and Data Science, National University of Singapore, Singapore.}
\thanks{Yuehaw Khoo is with the Department of Statistics, The University of Chicago, Chicago, IL 60637 USA.}
\thanks{Wanjie Wang is with the Department of Statistics and Data Science and the Department of Mathematics, National University of Singapore, Singapore.}
\thanks{Corresponding author: Wanjie Wang (e-mail: wanjie.wang@nus.edu.sg).}
}

\maketitle

\begin{abstract}
Robust point-set registration in the presence of noise and outliers is challenging because the matched points (inliers) must be identified before reliable alignment can be performed. Existing robust registration methods typically optimize over the transformation space and are often designed for regimes with a nonvanishing fraction of inliers. In this paper, we study the inlier recovery problem arising in robust registration by comparing two datasets through the Hadamard product of their Gram matrices. This formulation converts the inlier identification into a structured recovery problem and avoids direct optimization over the rotation group. Based on this idea, we develop two methods: an eigenvector matching method based on the leading eigenvector of the Gram-matrix overlap, and a row-sum matching method based on aggregated entrywise comparison. We show that the eigenvector method achieves weak recovery when the dimension and sample size are of the same order, while the row-sum method achieves exact recovery under a broader range of dimensional scalings. In particular, when the dimension is comparable to the sample size, exact recovery is possible even when the inlier fraction vanishes, with the number of inliers as small as order $\sqrt{n}$, up to logarithmic factors. We also discuss a parallel implementation for large-scale settings. Numerical experiments on brain imaging data and image examples demonstrate that the proposed methods effectively identify matched structure under substantial corruption.
\end{abstract}


\section{Introduction}
Rigid registration is a foundational problem in computer vision, robotics, and manifold learning, underpinning tasks ranging from multi-view geometry \cite{HartleyZisserman04} and Simultaneous Localization and Mapping \cite{DurrantWhyteBailey06,Cadena16}
to manifold alignment \cite{WangMahadevan08,Ovsjanikov12FMaps,GraveJoulinBerthet19}. In its classical form, one seeks a rotation (and sometimes a translation) that aligns two point sets. The problem of determining a rotation under a least-squares formulation penalizing misalignment admits a closed-form solution (see, e.g., Wahba's formulation and the SVD-based Procrustes solutions \cite{Wahba65,Arun87}). However, extending such a formulation to settings with grossly corrupted outliers is not straightforward: replacing the quadratic loss by robust alternatives \cite{verboon1994robust} for outlier rejection typically leads to non-convex optimization over the rotation group (with notable exceptions, such as the certifiable relaxations in \cite{yang2019polynomial}, albeit at a higher computational cost). Moreover, across the literature on robust point-set registration with outliers \cite{horowitz2014convex,bohorquez2020maximizing,yang2019polynomial,zhang2021fast}, it is typically shown—either theoretically or empirically—that one can tolerate a constant fraction of outliers under a specified noise model. It remains unclear whether one can go beyond this regime and handle a more adversarial setting in which the fraction of inliers is vanishingly small.

To cope with the difficulty of optimizing over the rotation group in robust registration and also answer to what extent we can tolerate outliers, we recast the task as a novel variant of a planted clique problem aimed at identifying the inlier points. By forming a similarity matrix within each point set, we seek the largest subset of indices on which the two similarity matrices exhibit maximal overlap. This connection brings robust registration into contact with a broad set of tools and analyses from graph inference, including planted clique and densest subgraph models \cite{AlonKrivelevichSudakov98,AriasCastroVerzelen14,HajekWuXu15}.

Our contributions are two-fold:
\begin{itemize}
\item We propose a spectral method that uses the leading eigenvector of a particular overlap matrix to identify the inlier points. Unlike the planted-clique literature, where i.i.d.-type noise is added directly to the graph adjacency matrix, our setting requires analyzing the Hadamard product of two Gram matrices. Since this matrix is fourth-order in the noise, the resulting perturbation exhibits heavier tails than the sub-gaussian or sub-exponential tails typically assumed in planted-model analyses. We show that when the number of nodes $n$ and the ambient dimension $d$ are of the same order, and the number of uncorrupted points is a constant fraction of $n$, the leading eigenvector is consistent with the indicator vector in $\ell_2$ norm (i.e., weak recovery \cite{li2021convex} in the planted model literature).
\item Establishing exact recovery (strong consistency) for the spectral method is more challenging, since it requires an $\ell_\infty$ eigenvector perturbation bound. To obtain strong consistency, in the spirit of \cite{AriasCastroVerzelen14}, we develop an alternative statistic based on row sums. A hypothesis test on the row-sum statistic allows us to derive exact recovery guarantees for the labels. With this method, we obtain a stronger result: when the number of nodes $n$ and the ambient dimension $d$ are of the same order, and the number of uncorrupted points is $\Omega(\sqrt{n})$, we achieve exact recovery of the inlier labels. Unlike the aforementioned spectral method, this approach can recover the labels even when the number of uncorrupted points grows much more slowly than the graph size. While results of this flavor are known in the planted clique literature, the classical planted clique model typically assumes i.i.d.\ sub-Gaussian noise added directly to the adjacency matrix, which is not the case here, where we have the Hadamard product of two Gram matrices. This result also suggests that robust registration can handle outliers beyond the typical constant-fraction corruption regime.
\end{itemize}

The remainder of the paper is organized as follows. Section~\ref{section:ModelSetting} introduces the model setting. In Section~\ref{section:eigMatch}, we propose a new eigenvector matching algorithm and establish convergence results for it. Section~\ref{section:rsMatch} presents a row-sum matching algorithm with theoretical guarantees, which relaxes the conditions on the dimension and sample size. A brief discussion of the parallel computing of our algorithms is provided in Section~\ref{section:parallel}. Section~\ref{section:numerical} presents numerical experiments, and Section~\ref{sec:case} provides a case study. Finally, Section~\ref{section:conclusion} concludes the paper by discussing the implications and possible extensions of our work.

\subsection{Notation:} For any matrix $M$, denote $M_i$ as the $i$-th column of $M$ and $M_{ij}$ as the $(i, j)$-th entry of $M$.  For two matrices $X$ and $Y$ with the same size, denote the Hadamard product as $X \circ Y$. Let $\zero_d=(0, \ldots, 0)^\top \in\reals^{d}$ and $\one_d = (1, \ldots, 1)^\top \in\reals^{d}$, where the subscript $d$ is omitted when unambiguous. Let $[n] = \{1,.\ldots, n\}$. 
Denote $\|v\|$ be the $\ell_2$ norm for a vector $v$ and $\|M\|$ be the spectral norm for a matrix $M$.
For any set $A$, use $|A|$ to represent the cardinality of the set. The constant $C$ and the 
$\gamma$ dependent constant $C_\gamma$ may change line to line.

\section{Model Setting}\label{section:ModelSetting}

\noindent Suppose we have two $d$-dimensional data point sets $\{X_i\}$ and $\{Y_i\}$, $i = 1, 2, \cdots, n$. 
Some points match each other up to a rotation, for which we call inliers. But some points are outliers that cannot be matched. 
Let $G$ denote the inlier set that $Y_i = R X_i$, where $R \in \mathbb{R}^{d\times d}$ is an orthogonal rotation matrix. Let $B$ denote the set of outliers, where $X_i$ and $Y_i$ are independent of each other. They might be irrelevant objects, or pure noise. The unmatched nodes in $B$ usually cause a big challenge in data registration \cite{yang2019polynomial, horowitz2014convex, bohorquez2020maximizing, zhang2021fast}. Therefore, it is of interest to identify $G$ and $B$, and data registration on only $G$ enjoys a satisfying result. 

We use $X=(X_1,...,X_n) \in\reals^{d\times n}$ and $Y=(Y_1,...,Y_n) \in\reals^{d \times n}$ to represent these two point sets. Each column of $X$ is a data point, and each row is a feature vector over all data points. We formulate the index sets $G$ and $B$ in Assumption \ref{aspt:setup}. 

\begin{aspt}\label{aspt:setup}
    To streamline the technical presentation, we assume the following:
    \begin{enumerate}
        \item The index sets $G \subset [n]$ and $B\subset [n]$ satisfies $G\bigcap B=\phi$ and $G\bigcup B=[n]$.
        \item Each data point follows $X_i \sim \mathcal{N}(\zero, I_d)$ independently. Furthermore, there exists an orthogonal matrix $R \in \mathbb{R}^{d \times d}$, so that 
        \[
        \left\{\begin{array}{ll}
        Y_i = RX_i, & i \in G\\
        Y_i \sim \mathcal{N}(\zero, I_d) \mbox{ independently }, & i \in B.
        \end{array}
        \right.
        \]
    \end{enumerate}
\end{aspt}

To differentiate $G$ and $B$, we exploit the Gram matrices for the two data point sets. Take $X$ as an example, the Gram matrix $X^\top X$ gives the similarities across node pairs, which is independent of the rotation $R$. Hence, for $i, j \in G$, the similarity between node pair $(i,j)$ is $X_i^\top X_j$ in $X^\top X$, and $Y_i^\top Y_j$ in $Y$. By $Y_i = RX_i$, we have that 
\[
Y_i^\top Y_j = X_i^\top R^\top R X_i = X_i^\top X_i.
\]
The two quantities are the same, even with the existence of rotation $R$. As a contrast, if $i \in B$ or $j \in B$, the similarities for $(i,j)$ differ much in $X^\top X$ and $Y^\top Y$. 

Based on this idea, we define the product matrix between these two Gram matrices:
\begin{equation}\label{equation:H}
    H = \left( X^\top X \right) \circ \left( Y^\top Y \right).
\end{equation}
The product matrix $H$ captures this comparison information. On the diagonals, $H_{ii} = (X_i ^\top X_i)(Y_i ^\top Y_i)$, which follows $\|X_i\|^4$ for $i \in G$ and $\|X_i\|^2 \|Y_i\|^2$ for $i \in B$. The former one has a larger magnitude than the latter due to the fourth moment. This difference is more obvious on the off-diagonals. When $i \neq j$, the entry $H_{ij} = (X_i^\top X_j)^2$ for $i, j \in G$ and  $E[H_{ij}] = 0$ if one of them is in $B$, due to the independent normal random variable.

Based on the different performance of entries in $H$, we develop two estimation methods: one is to employ the eigenvector information, and the other is to employ the row sums. In Section \ref{section:eigMatch}, we propose the eigenvector matching algorithm. This algorithm classifies entries of the leading eigenvector of $H$ into two sets to obtain the estimated sets $\hat{G}$ and $\hat{B}$. The consistency can be guaranteed when the dimension $d$ and the sample size $n$ are at the same divergence rate, i.e., $d/n =\gamma\in(0, \infty)$ for a constant $\gamma$. In Section \ref{section:rsMatch}, we propose the row sum matching approach. The consistency is guaranteed when $\log n \lesssim d \lesssim n^2/\log n$. We further develop the parallel computing algorithms in Section \ref{section:parallel} for both algorithms. 

Throughout this paper, we focus on the proportion of missing points in $G$ and $B$. For $G$, the error rate is the proportion of points in $G$ that are classified into $\hat{B}$, i.e., the false negatives. For $B$, $error_B(\hat{B})$ is the proportion of points classified into $\hat{G}$. The overall error rate is the proportion of misclassified points over the whole dataset. The definition is as follows:
\begin{defn}
    Let $\Ghat$ and $\Bhat$ be the estimated index sets of $G$ and $B$. The error rates on $G$ and $B$ are defined as 
    \begin{equation*}
        \eG(\Ghat) := \frac{|G \bigcap \Ghat^c|}{|G|}, \quad
        \eB(\Bhat) := \frac{|B \bigcap \Bhat^c|}{|B|}.
    \end{equation*}
    Similarly, define the overall error rate on the data as
    \begin{equation*}
        \eW(\Ghat, \Bhat) := \frac{|G \bigcap \Ghat^c| + |B \bigcap \Bhat^c|}{n}
    \end{equation*}
\end{defn}

\section{Eigenvector Matching} \label{section:eigMatch}

\subsection{Matching Algorithm with the Leading Eigenvector}

Consider the product matrix $H$ defined in Equation \eqref{equation:H}. Under Assumption \ref{aspt:setup}, the expectation $\E[H]$ exhibits a block structure induced by the partition $[n]=G\cup B$. To understand the information carried by this structure, we first examine the leading eigenvalues and eigenvectors of $\E[H]$.

Suppose Assumption \ref{aspt:setup} holds. Let $\lambda_1$ and $\lambda_2$ denote the largest and the second largest eigenvalues of $\E[H]$, respectively. Then
\begin{equation*}
    \lambda_1 = d^2 + d(|G|+1), \qquad \lambda_2 = d^2.
\end{equation*}
Hence, the eigengap is $\lambda_1-\lambda_2=d|G|=drn$, where $r=|G|/n\in(0,1)$. The leading eigenvector $v$ corresponding to $\lambda_1$ is given by
\begin{equation*}
    v_i=
    \begin{cases}
        1/\sqrt{rn}, & i\in G,\\
        0, & i\in B.
    \end{cases}
\end{equation*}
Therefore, at the population level, exact separation between $G$ and $B$ is achieved by thresholding $v$ at the level $tn^{-1/2}$ for any $t\in(0,r^{-1/2})$. This population-level observation motivates the empirical procedure below.

Based on this structure, we propose a matching algorithm using the leading eigenvector of the empirical matrix $H$; see Algorithm \ref{alg:Eigmatching}. We first construct $H$ from the two Gram matrices $X^\top X$ and $Y^\top Y$, then compute the leading eigenvector of $H$, and finally classify the indices into two groups. Ideally, when the empirical leading eigenvector is close to its population counterpart, its coordinates should inherit the same separation pattern.

\begin{algorithm}[htb]
\caption{Eigenvector Matching.}\label{alg:Eigmatching}
\begin{algorithmic}[1]
    \Require Data matrix $X \in \reals^{d\times n}$; Data matrix $Y \in \reals^{d\times n}$; Optional threshold $t$.
    \State Center each row of $X$ and $Y$, i.e. $X\boldsymbol{1} = \boldsymbol{0}$ and $Y\boldsymbol{1} = \boldsymbol{0}$. Then, normalise each column of $X$ and $Y$ so that $X = \left\{ \frac{X_1}{\|X_1\|},\ldots,\frac{X_n}{\|X_n\|} \right\}$ and $Y = \left\{ \frac{Y_1}{\|Y_1\|},\ldots,\frac{Y_n}{\|Y_n\|} \right\}$.
    \State Calculate $H= \left( X^\top X \right) \circ \left( Y^\top Y \right)$.
    \State Find the leading eigenvector $v\in\mathbb{R}^{n}$ of $H$.
    \If{$t$ is given,}
        \State Classify the indices into two sets with $\hat{G} = \{i\in[n]: v_i \geq tn^{-\frac{1}{2}}\}$ and $\hat{B} = \{i\in[n]: v_i < tn^{-\frac{1}{2}}\}$.
    \Else
        \State Apply $k$-means to $v$, treating each entry as a data point. Let $K = 2$ and initial centroids as $n^{-\frac{1}{2}}$ and $-n^{-\frac{1}{2}}$. Assign the cluster with the larger centroid as $\hat{G}$ and the cluster with the smaller centroid as $\hat{B}$.
    \EndIf
    \State \Return The estimated index sets $\hat{G}$ and $\hat{B}$.
\end{algorithmic}
\end{algorithm}

In practice, the proportion $r$ is unknown. If prior information on $r$ is available, then a natural choice is to select $t$ near $1/(2\sqrt r)$, or slightly smaller to be conservative; this choice is further justified by the theoretical error bound given later. If no prior information is available, then any fixed $t\in(0,1)$ always works, since $r\in(0,1)$ implies $r^{-1/2}>1$. This provides a simple universal choice. If one prefers not to tune the threshold parameter, then we recommend using $k$-means, which naturally separates the entries of the leading eigenvector into two groups according to their magnitudes. In the numerical analysis, we compare both thresholding and $k$-means, and both methods perform well.

The centring and normalisation steps in Algorithm \ref{alg:Eigmatching} are also important. They make the empirical behaviour of the columns of $X$ and $Y$ closer to that under the standard Gaussian model assumed in the theoretical analysis, and improve the stability of the subsequent eigenvector-based classification.

\subsection{Consistency Results of Eigenvector Matching}

We now analyse the performance of Algorithm \ref{alg:Eigmatching}. The argument has three steps: first, control $\|H-\E[H]\|$; second, derive an $\ell_2$ bound for the leading eigenvector of $H$; third, convert this bound into misclassification error rates for the thresholding rule. The result yields vanishing misclassification proportion, but not an $\ell_\infty$-type exact separation guarantee.

We begin with the perturbation of the product matrix $H$. This is nontrivial because the entries of $H$ are heavy-tailed and highly dependent, due to the combination of inner products and the Hadamard product. Therefore, standard random matrix results do not apply directly. We use an $\epsilon$-net argument; see Appendix \ref{appendix:theo_eig}. Related covering arguments can be found in \cite{Vershynin_2018}.

\begin{theorem} \label{theo:mainEig}
    Suppose Assumption \ref{aspt:setup} holds and $H$ is defined as in Equation \eqref{equation:H}. Suppose $n,d\to\infty$ with $d/n=\gamma\in(0,\infty)$. Then, with high probability $1-o(1)$,
    \begin{equation}
        \|H-\E[H]\| \leq C_\gamma n^{3/2}\log^2 n.
    \end{equation}
\end{theorem}

\begin{proof}
    See Appendix \ref{appendix:theo_eig}.
\end{proof}

This perturbation bound is small compared with the population eigengap. Indeed, the leading two eigenvalues of $\E[H]$ are $\lambda_1=d^2+d(rn+1)$ and $\lambda_2=d^2+d$, so $\lambda_1-\lambda_2=drn$, which is of order $n^2$ under $d/n=\gamma$. Hence, the leading eigenvector is stable under perturbation.

Applying the Davis--Kahan theorem \cite{Davis_Kahan} yields an $\ell_2$ bound for the leading eigenvector. This can then be translated into classification error bounds by thresholding the empirical eigenvector at the level $tn^{-1/2}$.

\begin{corollary} \label{cor:l2}
    Suppose Assumption \ref{aspt:setup} holds and $H$ is defined as in Equation \eqref{equation:H}. Let $n,d\to\infty$ with $d/n=\gamma\in(0,\infty)$. Denote by $v$ the leading eigenvector of $\E[H]$, and by $\vhat$ the leading eigenvector of $H$. Then, with high probability $1-o(1)$,
    \begin{equation*}
        \|\vhat-v\| \leq C_\gamma n^{-1/2}\log^2 n.
    \end{equation*}
    Furthermore, for any $t\in(0,r^{-1/2})$, let $\Ghat=\{i\in[n]: \vhat_i\ge tn^{-1/2}\}$ and $\Bhat=\{i\in[n]: \vhat_i<tn^{-1/2}\}$. Then
    \begin{align*}
        \eG(\Ghat) &\le \frac{C_\gamma \log^4 n}{n(1-t\sqrt r)^2}, \\
        \eB(\Bhat) &\le \frac{C_\gamma \log^4 n}{nt^2(1-r)}, \\
        \eW(\Ghat,\Bhat) &\le \frac{C_\gamma \log^4 n}{n\min\{(r^{-1/2}-t)^2,t^2\}}
    \end{align*}
    with high probability.
\end{corollary}
\begin{rem}
    The condition $t\in(0,r^{-1/2})$ is exactly the population separation condition. Since $r\in(0,1)$, any fixed $t\in(0,1)$ is admissible. Moreover, the upper bound for the overall error depends on $\min\{(r^{-1/2}-t)^2,t^2\}$, so the bound is minimized at $t=1/(2\sqrt r)$. This is the optimizer of the upper bound, not necessarily of the true finite-sample error.
\end{rem}
\begin{proof}
    By the Davis--Kahan theorem \cite{Davis_Kahan},
    \begin{equation*}
        \|\vhat-v\|\leq \frac{\sqrt{2}\|H-\E[H]\|}{|\lambda_1-\lambdahat_2|},
    \end{equation*}
    where $\lambdahat_2$ is the second largest eigenvalue of $H$, and $\lambda_1$ is the largest eigenvalue of $\E[H]$. Let $\lambda_2$ be the second largest eigenvalue of $\E[H]$. By Weyl's inequality, $\lambdahat_2 \leq \lambda_2 + \|H-\E[H]\|$, hence
    \begin{equation*}
        |\lambda_1-\lambdahat_2| \geq \lambda_1-\lambda_2-\|H-\E[H]\|.
    \end{equation*}
    Note that $\lambda_1 = d^2 + d(rn+1)$ and $\lambda_2 = d^2 + d$, so $\lambda_1-\lambda_2 = drn$. By Theorem \ref{theo:mainEig}, with high probability,
    \[
        \|H-\E[H]\| \leq C_\gamma n^{3/2}\log^2 n.
    \]
    Therefore,
    \begin{equation*}
        \|\vhat-v\| \leq \frac{C_\gamma n^{3/2}\log^2 n}{drn - C_\gamma n^{3/2}\log^2 n}.
    \end{equation*}
    Since $d/n=\gamma\in(0,\infty)$ and $r\in(0,1)$ is fixed, the denominator is of order $n^2$. Hence,
    \begin{equation*}
        \|\vhat-v\| \leq C_\gamma n^{-1/2}\log^2 n
    \end{equation*}
    with high probability $1-o(1)$.

    We now derive the error rates. For $i\in G\cap\Ghat^c$, the empirical coordinate satisfies $\vhat_i<tn^{-1/2}$, while the population coordinate is $v_i=(rn)^{-1/2}$. Since $t<r^{-1/2}$, we have $(\vhat_i-v_i)^2\ge n^{-1}(r^{-1/2}-t)^2$. Summing over $i\in G\cap\Ghat^c$ gives
    \begin{equation*}
        \|\vhat-v\|^2 \geq \frac{|G\cap\Ghat^c|}{n}(r^{-1/2}-t)^2.
    \end{equation*}
    Since $|G|=rn$, it follows that
    \[
        \eG(\Ghat)=\frac{|G\cap\Ghat^c|}{|G|}\le \frac{\|\vhat-v\|^2}{(1-t\sqrt r)^2},
    \]
    and therefore
    \begin{equation*}
        \eG(\Ghat)\le \frac{C_\gamma \log^4 n}{n(1-t\sqrt r)^2}
    \end{equation*}
    with high probability.

    For $i\in B\cap\Bhat^c$, we have $v_i=0$ and $\vhat_i\ge tn^{-1/2}$, so $(\vhat_i-v_i)^2\ge n^{-1}t^2$. Summing over $i\in B\cap\Bhat^c$ yields $\|\vhat-v\|^2 \ge n^{-1}|B\cap\Bhat^c|t^2$. Since $|B|=(1-r)n$, we obtain
    \begin{equation*}
        \eB(\Bhat)\le \frac{C_\gamma \log^4 n}{nt^2(1-r)}
    \end{equation*}
    with high probability.

    Combining the two bounds, we have
    \begin{equation*}
        \|\vhat-v\|^2 \geq \frac{|G\cap\Ghat^c|+|B\cap\Bhat^c|}{n}\min\{(r^{-1/2}-t)^2,t^2\},
    \end{equation*}
    which implies
    \begin{equation*}
        \eW(\Ghat,\Bhat)\le \frac{C_\gamma \log^4 n}{n\min\{(r^{-1/2}-t)^2,t^2\}}.
    \end{equation*}
\end{proof}

\section{Row Sum Matching} \label{section:rsMatch}

\subsection{Matching Algorithm with Row-sum}
To recover the matched and mismatched points, another method to summarize the differences over blocks is to aggregate the entries of each row in $H$. Hence, we exploit the row-sum statistic of $\E[H]$ and utilize it for classification. 

Assume Assumption \ref{aspt:setup} holds and let the product matrix $H$ be defined as in Equation \eqref{equation:H}. For each $i\in[n]$, define the row-sum statistic
\begin{equation*}
    S_i = e_i^\top H \boldsymbol{1}.
\end{equation*}
Since $H$ measures the overlap between the two Gram matrices, the row sums aggregate the matching information across all indices and provide a natural alternative to the leading-eigenvector approach.

At the population level, the row sums satisfy
\begin{equation*}
    \E[S_i] =
    \begin{cases}
        d^2 + d(rn+1), & i\in G,\\
        d^2, & i\in B.
    \end{cases}
\end{equation*}
Hence, the shifted statistic $S_i-d^2$ has expectation $d(rn+1)$ on $G$ and $0$ on $B$. Therefore, exact separation is achieved at the population level by thresholding $S_i-d^2$ at any level $T\in(0,d(rn+1))$.

Based on this observation, we propose the row-sum matching procedure in Algorithm \ref{alg:RSmatching}. After constructing $H$, we compute the row sums $S_i$ and classify the indices either by thresholding $S_i-d^2$ at a level $T$, or by applying $k$-means with $K=2$ to $\{S_i-d^2:i\in[n]\}$. If prior information on $r$ is available, then a natural choice is to take $T$ near $d(rn+1)/2$, or slightly smaller to be conservative. Otherwise, we recommend using $k$-means as a data-driven alternative. In the numerical analysis, we focus on the $k$-means version.

As in Algorithm \ref{alg:Eigmatching}, the centering and normalization steps are important in practice. They stabilize the empirical scale of the data matrices and improve the subsequent classification based on the row sums.

\begin{algorithm}[!htb]
\caption{Row Sum Matching.}\label{alg:RSmatching}
\begin{algorithmic}[1]
    \Require Data matrix $X \in \reals^{d\times n}$; Data matrix $Y \in \reals^{d\times n}$; Optional threshold $T$.
    \State Center each row of $X$ and $Y$, i.e., $X\boldsymbol{1} = \boldsymbol{0}$ and $Y\boldsymbol{1} = \boldsymbol{0}$. Then, normalize each column of $X$ and $Y$ so that $X = \left\{ \frac{X_1}{\|X_1\|},\ldots,\frac{X_n}{\|X_n\|} \right\}$ and $Y = \left\{ \frac{Y_1}{\|Y_1\|},\ldots,\frac{Y_n}{\|Y_n\|} \right\}$.
    \State Calculate $H= \left( X^\top X \right) \circ \left( Y^\top Y \right)$.
    \State Compute the row sums $S_i = e_i^\top H \boldsymbol{1}$ for $i\in[n]$.
    \If{$T$ is given,}
        \State Classify the indices into two sets with $\hat{G} = \{i\in[n]: S_i-d^2 \geq T\}$ and $\hat{B} = \{i\in[n]: S_i-d^2 < T\}$.
    \Else
        \State Apply $k$-means to $\{S_i-d^2: i\in[n]\}$, treating each entry as a data point. Let $K = 2$. Assign the cluster with the larger centroid as $\hat{G}$ and the other cluster as $\hat{B}$.
    \EndIf
    \State \Return The estimated index sets $\hat{G}$ and $\hat{B}$.
\end{algorithmic}
\end{algorithm}

\subsection{Consistency Results of Row Sum Matching}

We now establish the theoretical guarantees for the row-sum statistic. Compared with the eigenvector matching method, the row-sum analysis improves in three aspects. First, the condition on the dimension is relaxed from $d\asymp n$ to $\log n \lesssim d \lesssim n^2/\log n$. This means that the method applies not only when the dimension and sample size are of the same order, but also in both lower-dimensional and higher-dimensional regimes. Second, the requirement on the inlier proportion $r$ is weaker. In particular, under the same scaling $d\asymp n$ as in the eigenvector analysis, the row-sum method still permits $r\to 0$. Third, the conclusion is stronger: instead of only showing that the error rates vanish, we prove exact recovery by controlling the numbers of false negatives and false positives directly.

The proof is based on the row-sum statistics $S_i=e_i^\top H\boldsymbol{1}$, $i\in[n]$. Since $\E[S_i]=d^2+d(rn+1)$ for $i\in G$ and $\E[S_i]=d^2$ for $i\in B$, the two groups are separated at the population level by a gap of size $d(rn+1)$. Therefore, exact recovery follows once the fluctuations of $S_i$ around $\E[S_i]$ are uniformly smaller than this gap. The main technical task is thus to control $|S_i-\E[S_i]|$ separately for $i\in G$ and $i\in B$. This is carried out in Theorem~\ref{theo:mainRS}, whose proof is given in Appendix~\ref{appendix:theo_rs}.

\begin{theorem}\label{theo:mainRS}
    Let $H$ be defined as in Equation \eqref{equation:H} and suppose Assumption \ref{aspt:setup} holds. Denote $S_i=e_i^\top H\boldsymbol{1}$. Let $n,d\to\infty$ with $\log n \lesssim d \lesssim n^2/\log n$. Then there exists $n_0$ such that for all $n>n_0$, the following statements hold:
    
    \begin{enumerate}
        \item[i)] For $i\in G$, with probability at least $1-9n^{-2}$,
        \begin{equation*}
            |S_i-\E[S_i]| \leq C_1\sqrt{d\log n}\left(d+\sqrt{dn}+rn\right).
        \end{equation*}
        
        \item[ii)] For $i\in B$, with probability at least $1-10n^{-2}$,
        \begin{equation*}
            |S_i-\E[S_i]| \leq C_2 d\sqrt{\log n}\left(\sqrt d+\sqrt n\right).
        \end{equation*}
    \end{enumerate}
\end{theorem}

Theorem \ref{theo:mainRS} shows that the row-sum analysis is valid throughout the regime $\log n \lesssim d \lesssim n^2/\log n$. In particular, it covers both the low-dimensional case $d\asymp \log n$ and the high-dimensional case $d\asymp n^2/\log n$, whereas the eigenvector analysis requires $d\asymp n$.

With these perturbation bounds in hand, we can identify a threshold $T$ that separates the two groups exactly. Indeed, if $T$ lies between the upper fluctuation level of the $B$-indexed row sums and the lower fluctuation level of the $G$-indexed row sums, then every index is classified correctly. This is formalized in Corollary~\ref{cor:RS_rate}.

\begin{corollary}\label{cor:RS_rate}
    Let $H$ be defined as in Equation \eqref{equation:H} and suppose Assumption \ref{aspt:setup} holds. Denote $S_i=e_i^\top H\boldsymbol{1}$. Let $n,d\to\infty$ with $\log n \lesssim d \lesssim n^2/\log n$. Define
    \[
        \Ghat=\{i\in[n]: S_i-d^2\ge T\}, \qquad
        \Bhat=\{i\in[n]: S_i-d^2<T\}.
    \]
    If $T$ is chosen in the interval
    \begin{equation}\label{equation:condT}
        C_2 d\sqrt{\log n}\left(\sqrt d+\sqrt n\right)
        <
        T
        <
        d(rn+1)-C_1\sqrt{d\log n}\left(d+\sqrt{dn}+rn\right),
    \end{equation}
    then there exists $n_0$ such that for all $n>n_0$, with probability at least $1-19n^{-1}$, the number of errors is zero, that is,
    \[
        |G\cap \Ghat^c|=0, \qquad |B\cap \Bhat^c|=0.
    \]
\end{corollary}

\begin{proof}
    By Theorem \ref{theo:mainRS}, there exists $n_0$ such that for all $n>n_0$,
    \[
        P\!\left(|S_i-\E[S_i]|\ge C_1\sqrt{d\log n}\left(d+\sqrt{dn}+rn\right)\right)\le 9n^{-2},
        \qquad i\in G,
    \]
    and
    \[
        P\!\left(|S_i-\E[S_i]|\ge C_2 d\sqrt{\log n}\left(\sqrt d+\sqrt n\right)\right)\le 10n^{-2},
        \qquad i\in B.
    \]

    Applying the union bound over all $i\in G$ yields
    \[
        P\!\left(\max_{i\in G}|S_i-\E[S_i]|\ge C_1\sqrt{d\log n}\left(d+\sqrt{dn}+rn\right)\right)\le 9n^{-1},
    \]
    and similarly, applying the union bound over all $i\in B$ gives
    \[
        P\!\left(\max_{i\in B}|S_i-\E[S_i]|\ge C_2 d\sqrt{\log n}\left(\sqrt d+\sqrt n\right)\right)\le 10n^{-1}.
    \]

    Choose $T$ in the stated interval. Then, with probability at least $1-19n^{-1}$, we have simultaneously
    \[
        \min_{i\in G}(S_i-d^2-T)\ge 0
        \qquad\text{and}\qquad
        \min_{i\in B}(T-S_i+d^2)> 0.
    \]
    Hence, every $G$-indexed row sum lies above the threshold and every $B$-indexed row sum lies below it. Therefore,
    \[
        |G\cap \Ghat^c|=0, \qquad |B\cap \Bhat^c|=0,
    \]
    which proves the claim.
\end{proof}

To ensure that the admissible interval for $T$ is nonempty, the signal term $d(rn+1)$ must dominate the fluctuation terms. This imposes an implicit lower bound on the inlier proportion $r$ as a function of both $d$ and $n$. At the level of order, this condition can be summarized as $rn \gtrsim \sqrt{d\log n}+\sqrt{n\log n}$.

This condition has a clear interpretation in different dimensional regimes. If $d\asymp n$, which is the same scaling condition as in the eigenvector matching analysis, then the above requirement reduces to $rn \gtrsim \sqrt{n\log n}$. In particular, the inlier proportion $r$ is allowed to vanish with $n$, and the number of inliers $|G|=rn$ can be as small as order $\sqrt n$, up to logarithmic factors. By contrast, if $d$ is very large, for example near the upper end $d\asymp n^2/\log n$, then $\sqrt{d\log n}\asymp n$, and the condition becomes $rn\gtrsim n$. Therefore, in this challenging high-dimensional regime, the inlier proportion $r = \Omega(1)$.

\begin{rem}
    Corollary 2 yields an exact recovery result under the broad dimensional regime $\log n \lesssim d \lesssim n^2/\log n$. In practice, when the threshold $T$ is unavailable, Algorithm \ref{alg:RSmatching} uses $k$-means as a data-driven alternative.
\end{rem}

\section{Parallel Computing}\label{section:parallel}

When the sample size $n$ is large, both the construction of the product matrix $H$ and the subsequent classification step can become computationally expensive, in terms of both memory usage and running time. This is particularly pronounced for the eigenvector method, which requires an eigenvalue decomposition of an $n\times n$ matrix, and also for the row-sum method, whose main computation is quadratic in $n$. Since both methods are based on pairwise comparisons within the sample, a natural way to reduce the computational burden is to split the full data into several smaller subsamples, process them in parallel, and then aggregate the estimated labels.

Consider first the eigenvector matching method. Suppose the full sample is split into $s$ subsamples of size roughly $m$, so that $n\approx ms$. For the subsample problem, the same eigenvector analysis applies as long as $m$ and $d$ remain of the same order. In this setting, the computational cost of a full eigenvalue decomposition is reduced from $O(n^3)$ to $O(sm^3)$. If the leading eigenvector is computed by an iterative method with $k$ iterations, then the cost is reduced from $O(kn^2)$ to $O(skm^2)$.

The same strategy applies to the row-sum method. After splitting the data into $s$ subsamples, we compute the row sums within each subsample and classify the indices locally. The row-sum computation is quadratic in the subsample size, so the total cost decreases from $O(n^2)$ to $O(sm^2)$. For the threshold-based theoretical guarantee, the subsample size $m$ should continue to satisfy the corresponding condition in Corollary \ref{cor:RS_rate}, so that the admissible interval for the threshold remains nonempty.

This parallel computing strategy is summarized in Algorithm \ref{alg:parallel}. In Section \ref{section:numerical}, we illustrate through numerical simulations that Algorithm \ref{alg:parallel} significantly decreases computational time while maintaining strong performance for both matching methods.

\begin{algorithm}[htb]
\caption{Matching via Parallel Computing.}\label{alg:parallel}
\begin{algorithmic}[1]
    \Require Data matrix $X \in \reals^{d\times n}$; Data matrix $Y \in \reals^{d\times n}$; Number of splits $s$; Method `algo'; Optional threshold parameter.
    \State Center each row of $X$ and $Y$, i.e. $X\boldsymbol{1} = \boldsymbol{0}$ and $Y\boldsymbol{1} = \boldsymbol{0}$. Then, normalize each column of $X$ and $Y$ so that $X = \left\{ \frac{X_1}{\|X_1\|},\ldots,\frac{X_n}{\|X_n\|} \right\}$ and $Y = \left\{ \frac{Y_1}{\|Y_1\|},\ldots,\frac{Y_n}{\|Y_n\|} \right\}$.
    \State Randomly split $X$ into $s$ subsamples with similar sizes. Denote the $j$-th subsample as $X^{(j)} \in \reals^{d\times n_j}$. Let $Y^{(j)} \in \reals^{d\times n_j}$ be the corresponding subsample of $Y$.
    \If{`algo' == `Eig'}
        \State For each server $1\le j\le s$, send $X^{(j)}$, $Y^{(j)}$, and the optional threshold $t$ to server $j$, and apply Algorithm \ref{alg:Eigmatching} to obtain $\hat G_j$ and $\hat B_j$.
    \ElsIf{`algo' == `RS'}
        \State For each server $1\le j\le s$, send $X^{(j)}$, $Y^{(j)}$, and the optional threshold $T$ to server $j$, and apply Algorithm \ref{alg:RSmatching} to obtain $\hat G_j$ and $\hat B_j$.
    \EndIf
    \State Set $\hat G = \varnothing$ and $\hat B = \varnothing$.
    \For{$j=1$ to $s$}
        \State $\hat G \leftarrow \{\hat G,\hat G_j\}$
        \State $\hat B \leftarrow \{\hat B,\hat B_j\}$
    \EndFor
    \State \Return The estimated index sets $\hat G$ and $\hat B$.
\end{algorithmic}
\end{algorithm}

\section{Numerical Results}\label{section:numerical}

In this section, we report three numerical experiments. The first two are based on the brain dataset and are designed to study statistical sensitivity to the inlier proportion and computational scalability under parallelization, respectively. The third uses the sky dataset to examine sensitivity to noise.

\subsection{Brain Data: Sensitivity to Inlier Proportion}\label{subsec:expt_ratio}

\noindent We first examine how the two matching methods respond to the amount of matching structure in the data. This experiment is based on the brain dataset and studies the sensitivity of the classification performance to the inlier proportion $r=|G|/n$.

We use brain imaging data from the Human Connectome Project. Details of the data acquisition and pre-processing are described in \cite{Brain_overview, Brain_diffusion, Brain_preprocessing1, Brain_preprocessing2, Brain_preprocessing3, Brain_preprocessing4, Brain_myelin, Brain_general}. The dataset is taken from the Connectome Workbench v1.5 Tutorial. Our analysis focuses on the left hemisphere, which contains $32{,}488$ vertices. For each vertex, we use six features: a three-dimensional spatial coordinate ({\it midthickness}), cortical thickness ({\it corrThickness}), curvature ({\it curvature}), and a proxy for myelin content ({\it MyelinMap}). We further pre-process the data by removing outliers in {\it MyelinMap}; see Appendix~\ref{appendix:Brain} for details.

To generate the matching problem, we take the data matrix $X \in \reals^{d \times n}$ from $n$ sampled vertices in the brain image data. Let $G$ denote the correctly matched set and $B$ the mismatched set. Given an orthogonal matrix $R \in \reals^{d \times d}$, we construct $Y$ by setting $Y_i = RX_i$ for $i \in G$, while for $i \in B$ we set $Y_i = RX_{\pi(i)}$ for an arbitrary permutation $\pi: B \to B$. Hence the two data matrices contain the same observations up to rotation, and the mismatch is created entirely by incorrect pairing on the set $B$. We fix $n=2000$ and vary the inlier proportion over $r \in \{0.55, 0.60, \ldots, 0.90\}$.

We compare five classification methods: Eigenvector Matching (Algorithm~\ref{alg:Eigmatching}) with thresholds $t \in \{0.3, 0.5, 0.7\}$, Eigenvector Matching with $k$-means, and Row Sum Matching (Algorithm~\ref{alg:RSmatching}) with $k$-means. For each value of $r$, we repeat the experiment 1000 times and report the average error rates in Figure~\ref{fig:Experiment1_RatioChange}.

As $r$ increases, all methods achieve lower error rates on $G$, $B$, and the whole dataset. This is consistent with the theoretical picture: a larger inlier proportion enlarges the signal gap and makes the matching problem easier. For a fixed $r$, a larger threshold in the eigenvector method decreases the error on $G$ but increases the error on $B$, reflecting the trade-off already suggested by the theoretical analysis.

For the overall error rate, the choice $t=0.5$ performs best among the three fixed thresholds. The data-driven $k$-means version of the eigenvector method performs similarly to the case $t=0.3$, and it gives the best performance when $r=0.9$. The row-sum method with $k$-means performs comparably to the eigenvector method with $k$-means throughout the experiment. Overall, these results show that both methods are sensitive to the amount of matching structure in a predictable way, and that the $k$-means versions provide competitive performance without requiring manual threshold tuning.

\begin{figure}[!t]
    \centering
    \includegraphics[width=\textwidth]{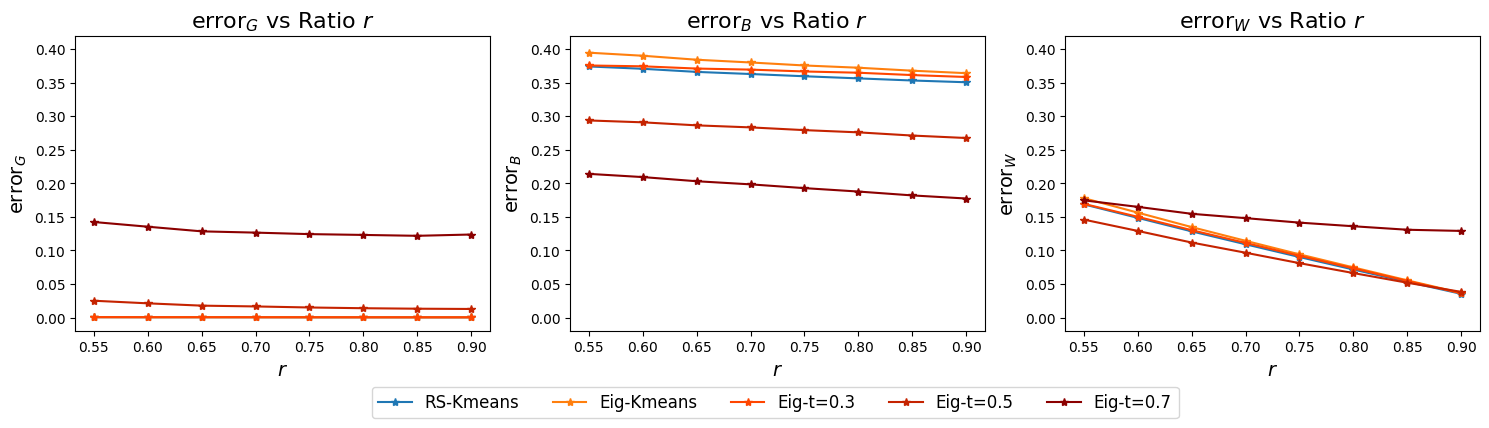}
    \caption{Performance across different inlier proportions on the brain dataset, with $n=2000$ and $r \in \{0.55, 0.6, 0.65, 0.7, 0.75, 0.8, 0.85, 0.9\}$. We compare Eigenvector Matching (Algorithm~\ref{alg:Eigmatching}) with $k$-means and thresholds $t \in \{0.3, 0.5, 0.7\}$, and Row Sum Matching (Algorithm~\ref{alg:RSmatching}) with $k$-means. Left panel: $\eG$ versus $r$; middle panel: $\eB$ versus $r$; right panel: $\eW$ versus $r$.}
    \label{fig:Experiment1_RatioChange}
\end{figure}

\subsection{Brain Data: Parallel Computing Performance}\label{subsec:expt_parallel}

\noindent We next examine the computational performance of the parallel computing strategy proposed in Section~\ref{section:parallel}. This experiment is again based on the brain dataset and is designed to evaluate how parallelization affects both running time and classification accuracy.

The data generation mechanism is the same as in Section~\ref{subsec:expt_ratio}. We use the brain data to construct a matching problem with an inlier proportion fixed at $r=0.8$, and we take $n=12000$. We then split the full sample into $s \in \{1,2,3,4,6,12\}$ subsamples and apply the parallel version of the algorithm in each case. For the eigenvector method, we use the $k$-means version of Algorithm~\ref{alg:Eigmatching}; for the row-sum method, we use the $k$-means version of Algorithm~\ref{alg:RSmatching}. The computations are carried out on a 13th-generation Intel Core i9-13900HX processor.

The average error rates and total computation times over 1000 repetitions are summarized in Table~\ref{table:Eig} for the eigenvector method and Table~\ref{table:RS} for the row-sum method. To keep the presentation concise, we report only the mean and standard deviation of $\eW$; the behaviors of $\eG$ and $\eB$ are similar.

For both methods, the error rates remain essentially unchanged as the number of splits increases. This indicates that the parallel strategy preserves the statistical performance of the original procedures. At the same time, the running time decreases substantially when the data are split into a moderate number of subsamples. In particular, moving from no split to two splits reduces the running time by about $30\%$ for both methods.

The computational gain becomes less pronounced when the number of splits is further increased. For example, the reduction in running time from six splits to twelve splits is relatively small. This is expected: although smaller subsamples reduce the cost of matrix construction, the classification step still incurs a non-negligible overhead. Overall, the experiment confirms that parallelization provides a practical improvement in computation while maintaining stable classification accuracy.

\begin{table}
    \caption{Running time and performance with Eigenvector Matching.\label{table:Eig}}
    \centering
    \begin{tabular}{ |c|c|c| }
        \hline
        No. of Split & Running Time & $\eW$ \\
        \hline
        no split: $n_s=12000$ & 2922s & 0.075 (0.002) \\
        2 split: $n_s=6000$ & 1949s & 0.075 (0.002) \\
        3 split: $n_s=4000$ & 1787s & 0.075 (0.002) \\
        4 split: $n_s=3000$ & 1701s & 0.075 (0.002) \\
        6 split: $n_s=2000$ & 1621s & 0.075 (0.002) \\
        12 split: $n_s=1000$ & 1560s & 0.075 (0.002) \\
        \hline
    \end{tabular}
\end{table}

\begin{table}
    \caption{Running time and performance with Row Sum Matching.\label{table:RS}}
    \centering
    \begin{tabular}{ |c|c|c| }
        \hline
        No. of Split & Running Time & $\eW$ \\
        \hline
        no split: $n_s=12000$ & 2701s & 0.072 (0.002) \\
        2 split: $n_s=6000$ & 1846s & 0.072 (0.002) \\
        3 split: $n_s=4000$ & 1693s & 0.072 (0.002) \\
        4 split: $n_s=3000$ & 1635s & 0.072 (0.002) \\
        6 split: $n_s=2000$ & 1577s & 0.072 (0.002) \\
        12 split: $n_s=1000$ & 1499s & 0.072 (0.002) \\
        \hline
    \end{tabular}
\end{table}

\subsection{Sky Data: Sensitivity to Noise}\label{subsection:sky}

\noindent We next examine how the two matching methods respond to increasing noise. This experiment is based on the sky dataset and is designed to study the sensitivity of the classification performance to the noise level.

We use the RGB values of pixels from a sky image to construct the data matrices. Each pixel is represented by a three-dimensional vector corresponding to its red, green, and blue values. We consider two perturbation scenarios:
\begin{itemize}
    \item Scenario 1: Two interchanged areas. Select two $900\times900$ squares on the original sky photo and interchange the RGB values of the selected two regions. The interchanged areas take $25\%$ of the total pixels.
    \item Scenario 2: Multiple permuted areas. Select 30 non-overlapping $200\times200$ squares on the original sky photo, and name the selected areas $S_1,\ldots,S_{30}$. The selected $200\times200$ blocks take $18.5\%$ of the total pixels.
\end{itemize}
In this section, we treat the original image as the data matrix $X$ of size $3\times6489600$ and the interchanged (or permuted) photo contaminated by random noise as the data matrix $Y$. The columns of the noise matrix follow a multivariate Gaussian distribution $\mathcal{N}(\zero,\sigma^2 I)$. We visualize the Scenario 1 and 2 sky photos in Figure~\ref{fig:Experiment3_Sky}. The photos are displayed on the $130\times195$ scale for better visualisation of the noise effect.

\begin{figure}[!t]
    \centering
    \includegraphics[width=\linewidth]{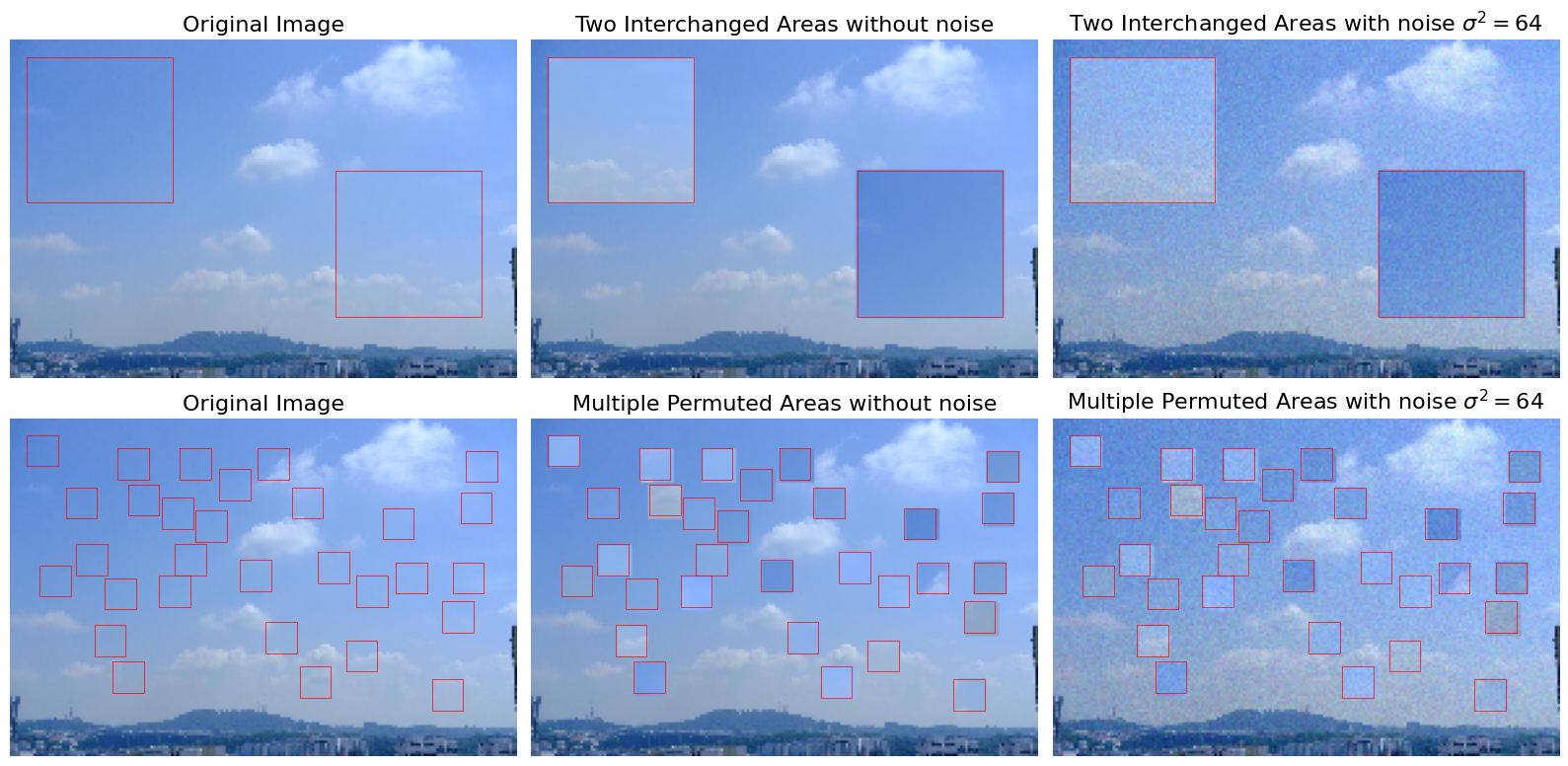}
    \caption{From left to right, we show the original sky photo, the changed sky photo, and the changed sky photo contaminated with noise. The upper panel shows Scenario 1, and the lower panel shows Scenario 2. The red box highlights the areas of exchange.}
    \label{fig:Experiment3_Sky}
\end{figure}

The data generation strategy is as follows. Randomly draw $n$ pixels to form $X \in \reals^{3\times n}$ from the original photo. For Scenario 1, let $Y \in \reals^{3\times n}$ be the corresponding RGB data of $X$ in the interchanged and polluted photo. For Scenario 2, we randomly permute the location of the squares $S_1,\ldots,S_{30}$ before adding noise, and let $Y \in \reals^{3\times n}$ be the corresponding RGB data of $X$ in the changed photo. We take $n=5000$ and vary the noise level over $\sigma^2 \in \{0,5,\ldots,100\}$. We compare Eigenvector Matching (Algorithm~\ref{alg:Eigmatching}) with thresholds $t \in \{0.3,0.5,0.7\}$, Eigenvector Matching with $k$-means, and Row Sum Matching (Algorithm~\ref{alg:RSmatching}) with $k$-means. For each value of $\sigma^2$, the experiment is repeated 1000 times, and the average error rates are reported in Figure~\ref{fig:Experiment3_Compare}.

\begin{figure}[!t]
    \centering
    \includegraphics[width=0.67\linewidth]{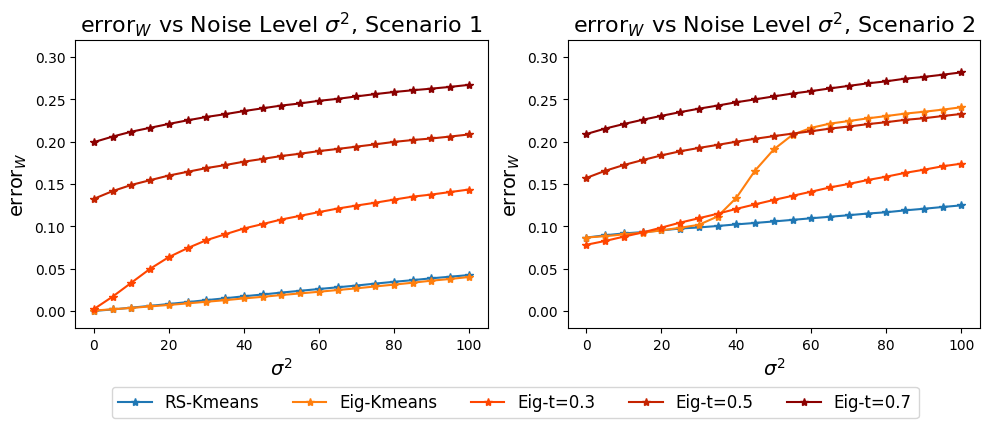}
    \caption{Performance as the noise level increases for the two sky-data scenarios. Left panel: Scenario 1; right panel: Scenario 2.}
    \label{fig:Experiment3_Compare}
\end{figure}

For both scenarios, the overall classification error increases as the noise level becomes larger. This is expected, since stronger noise reduces the separation between matched and mismatched pixels and therefore makes the classification task more difficult. For any fixed noise level, the eigenvector method with a larger threshold produces a larger overall error, indicating that in this experiment, a small positive threshold is more effective for separating matched and mismatched pixels.

In Scenario 1, the $k$-means versions of the eigenvector method and the row-sum method perform similarly, and both outperform the fixed-threshold eigenvector rules. In Scenario 2, the row-sum method with $k$-means gives the best performance among all methods, especially at intermediate noise levels.

Overall, both methods degrade as the noise level increases, while the row-sum method appears more stable in the more complex perturbation setting.

\section{Case Study}\label{sec:case}

In this section, we illustrate the practical behavior of the proposed methods on a simple image-difference detection task. The goal is to identify local regions that appear in one image but not in another, and to examine whether the two matching procedures can correctly localize these differences.

We consider three photographs of the same corner of a room: the original scene, the same scene with an extra yellow box, and the same scene with an extra yellow box and an orange pillow. Based on these images, we study three comparison tasks: detecting the extra box relative to the original image, detecting both extra items relative to the original image, and detecting the extra pillow relative to the image that already contains the box.

For the first two tasks, we compare the original photo with the photo containing the extra box and with the photo containing both the box and the pillow, respectively. For the third task, we compare the photo containing the extra box with the photo containing both extra items, so that the target difference is the additional pillow only. In all three cases, the goal is to identify the pixels corresponding to the newly introduced objects.

\begin{figure}[!t]
    \centering
    \includegraphics[width=\linewidth]{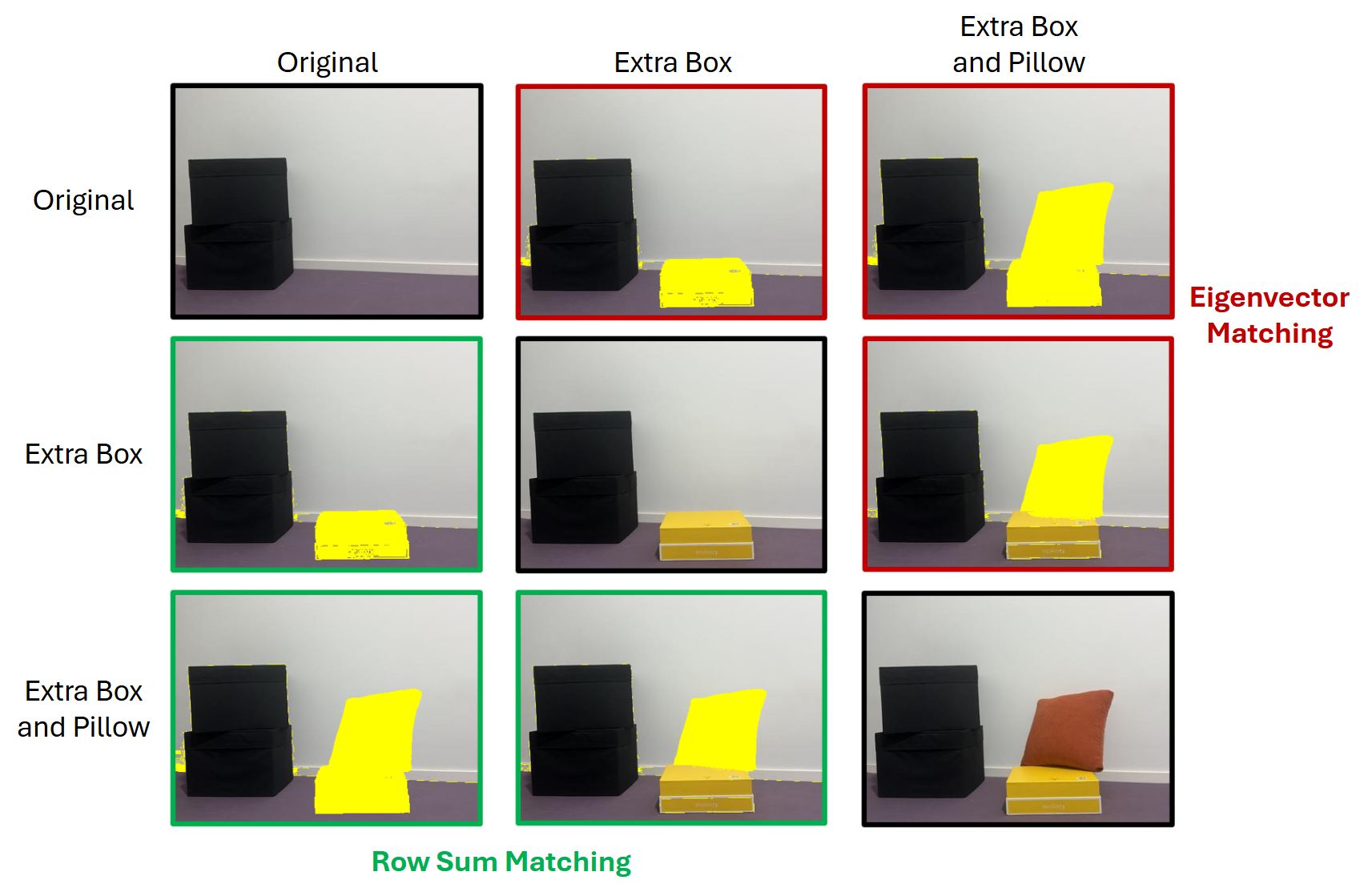}
    \caption{The yellow-highlighted pixels are the ones identified by the algorithm as extra items. The plots with a black frame on the diagonal are the photos used in comparison. The plots with a green frame use the row sum matching algorithm, and the plots with a red frame use the eigenvector matching algorithm.}
    \label{fig:Experiment4_Comparison}
\end{figure}

Figure~\ref{fig:Experiment4_Comparison} displays the results. The images on the diagonal are the original inputs used in each comparison. The panels framed in red show the output of Eigenvector Matching, while the panels framed in green show the output of Row Sum Matching. In each case, the yellow-highlighted pixels indicate the regions identified as different by the algorithm.

For both methods, the differing regions are detected correctly in all three comparisons. In particular, the added box and pillow are clearly localized. At the same time, some pixels in the lower-left part of the images are repeatedly identified as different. This is likely caused by changes in illumination and shadow: when a bright object is added to the scene, nearby shadowed regions become visually darker, and the algorithms interpret these intensity changes as additional differences.

Overall, this case study shows that both matching methods can be used to identify local changes between real images, even when the differing region is relatively small compared with the full image. The results also illustrate a practical limitation: visual changes caused by lighting and shadow may be detected together with the actual inserted objects.

\section{Conclusion}\label{section:conclusion}

In this paper, we studied robust registration with outliers through the product matrix $H$, which transforms the problem into one of label recovery for the inlier set. Based on this formulation, we proposed two matching methods: an eigenvector-based method and a row-sum-based method.

The two methods play complementary roles. The eigenvector method provides a spectral approach with weak recovery guarantees, while the row-sum method yields exact recovery under a broader range of dimensional regimes. Together, they show that accurate inlier identification is possible even under substantial corruption.

We also developed a parallel implementation to improve scalability on large datasets. Numerical experiments and the case study demonstrate that the proposed methods perform well in practice and can be used to identify meaningful differences between paired datasets.

Overall, the results suggest that robust registration can be fruitfully analyzed from the perspective of weighted label recovery, and that this viewpoint leads to both useful theory and effective algorithms.

\section*{Acknowledgments}
Data in Section \ref{subsec:expt_ratio} and \ref{subsec:expt_parallel} were provided by the Human Connectome Project, WU-Minn Consortium (Principal Investigators: David Van Essen and Kamil Ugurbil; 1U54MH091657) funded by the 16 NIH Institutes and Centers that support the NIH Blueprint for Neuroscience Research; and by the McDonnell Center for Systems Neuroscience at Washington University.

\bibliographystyle{plain}
\bibliography{bib}


\appendix

\section{Pre-processing on the Brain Image Data}\label{appendix:Brain}

The details of the data features and pre-processing strategies are listed in Table \ref{table:BrainInfo}.

\begin{table}[!htb]
\footnotesize
    \caption{Description of the Brain Data and Pre-processing Strategy}
    \label{table:BrainInfo}
    \centering
    \begin{tabular}{|c|c|c|}
        \hline
        Features & Details & Data Preparation \\
        \hline
        midthickness & 3d positional information & scale the coordinates to make the features comparable \\
         & & (after centralisation, value is ranged between [-3, 3]) \\
        \hline
        corrThickness & thickness of the cortex & $\backslash$ \\
        \hline
        curvature & degree of local folding & $\backslash$ \\
        \hline
        MyelinMap & proxy of myelin content & omit the extreme values according to its $99.99\%$ quantile. \\
        \hline
    \end{tabular}
\end{table}

\section{Preliminary Lemmas}\label{appendix:pre_lemmas}

\begin{lem}
    \label{lemma:AX}
    Suppose $X_1,...,X_n$ follows $\mathcal{N}(\zero, I_d)$ independently. The condition $n, d \rightarrow \infty$ with $d/n = \gamma \in(0, \infty)$ is satisfied. Then, with probability at least $1 - C_\gamma d^{-2}$, the following holds:
    \begin{equation*}
        \max_{i\neq j} \left| X_i^\top X_j \right| \leq C \sqrt{d\log{d}}.
    \end{equation*}
\end{lem}

\begin{proof}
    Note that $\|X_i\|^2 \sim \mathcal{X}^2_d$. By the tail probability of the Chi-square random variable, we can derive $\bbP\left( \|X_i\|^2 \geq 3d \right) \leq d^{-4}$. For each pair of $i \neq j$, we decompose:
    \begin{align*}
        \bbP \left( \left| X_i^\top X_j \right| \geq C \sqrt{d\log{d}} \right)
        \leq & \bbP \left( \left| X_i^\top X_j \right| \geq C\sqrt{d\log{d}}, \|X_i\| \leq \sqrt{3d} \right) + \bbP \left( \|X_i\|^2 \geq 3d \right) \\
        \leq & \bbP\left( \left| X_i^\top X_j \right| \geq C \sqrt{\log{d}} \|X_i\| \right) + d^{-4}
    \end{align*}
    Since $\frac{X_i^\top X_j}{\|X_i\|} \sim \mathcal{N}(0, 1)$, the tail probability satisfies
    \begin{equation*}
        \bbP \left( \left| \frac{X_i^\top X_j}{\|X\|} \right| \geq C\sqrt{\log{d}} \right) \leq 2d^{-4}.
    \end{equation*}
    Then, $\bbP\left( \left| X_i^\top X_j \right| \geq C\sqrt{d\log{d}} \right) \leq 3d^{-4}$ follows. Note that there are $n(n-1)/2$ distinct pairs of $i$ and $j$, the uniform bound can be easily derived, 
    \begin{equation*}
        \bbP\left( \max_{i\neq j} \left| X_i^\top X_j \right| \geq C\sqrt{d\log{d}} \right) 
        \leq \frac{n(n-1)}{2} \cdot 3d^{-4} 
        \leq C_\gamma d^{-2}.
    \end{equation*}
    The last equality follows as $d/n=\gamma$.
\end{proof}

\begin{lem}
    \label{lemma:BB}
    Suppose $X=(X_1,...,X_n)\in\reals^{d\times n}$ and $Y=(Y_1,...,Y_n)\in\reals^{d\times n}$ are independent matrices. The collection $\{ X_1,\ldots, X_n, Y_1,\ldots,Y_n \}$ consists of independent random vectors, each distributed as $\mathcal{N}(\zero, I_d)$. Let $u\in \reals^n$ satisfy $\|u\|=1$, and define the function $F_u:\reals^{d\times n} \times \reals^{d\times n}\rightarrow \reals$ by
    \begin{equation*}
        F_u(X, Y) = \sum_{i\in [n]} \sum_{j\in[n]\backslash\{i\}} u_i u_j \left( X_i^\top X_j \right) \left( Y_i^\top Y_j \right).
    \end{equation*}
    Also, define the event $A_x$ as
    \begin{equation*}
        \max_{i\neq j} \left| X_i^\top X_j \right| \leq C \sqrt{d\log{d}}, \quad \|X\| \leq C\sqrt{n\log{n}}.
    \end{equation*}
    Assume $A_X$ takes place and $n, d \rightarrow \infty$ with $d/n = \gamma \in (0, \infty)$. Then, for any fixed constant $C > 0$, there is an $C_\gamma > 0$ such that 
    \begin{equation*}
        \bbP\left( \left. F_u(X, Y) \geq C n^{\frac{3}{2}}\log^2{n} \right| X \right) \leq 3\exp\{ -C_\gamma n\log{n} \}.
    \end{equation*}
\end{lem}

\begin{proof}
    Denote $\|Y\| \leq C\sqrt{n\log{n}}$ as event $B_Y$. Columns of $Y$ are independent standard Gaussian random vectors in $\reals^d$, by \cite{Vershynin_2018}, $\bbP\left( B_Y \right) \geq 1-2\exp\{ -C_0n\log{n} \}$.

    Suppose $\tilde{Y}$ is an independent copy of $Y$. Define $\phi(y) = \E[\exp\{ t(F_u(X,Y) - y) \} | X]$. Let $y = F_u(X, \tilde{Y})$, then $\E[y|X]=0$. By Jensen's inequality,
    \begin{align} \label{eq:jensenBB}
        \E\left[ \left. \exp\left\{ tF_u(X, Y) \right\} \right| X \right]  
        \leq \E\left[ \left. \exp\left\{ t \left( F_u(X, Y) - F_u(X, \tilde{Y}) \right) \right\} \right| X \right].
    \end{align}
    Next, we check
    \begin{equation*}
        \nabla_{Y_i}F_u(X, Y) = 2 \sum_{j\in [n]\backslash\{i\}} u_i u_j(X_i^\top X_j) Y_j = u_i Y_{-i} V_i,
    \end{equation*}
    where $Y_{-i}$ is the $Y$ matrix with the $i$-th column taking zeros, and $V_i\in\reals^n$ with entries: $V_{i,j} = 2 u_j(X_i^\top X_j)$, $j\in[n]\backslash\{i\}$; $V_{i,i} = 0$. When $A_X$ and $B_Y$ takes place,
    \begin{align} \label{eq:normBB}
        \left\| \nabla_{Y} F_u(X, Y) \right\|^2 
        \leq \|Y\|^2 \sum_{i\in[n]} u_i^2 \|V_i\|^2 
        = \|Y\|^2 \sum_{i\in[n]} \sum_{j\in[n]\backslash\{i\}} 4 u_i^2 u_j^2 (X_i^\top X_j)^2 
        \leq Cnd\log{n} \log{d} .
    \end{align}
    Define $Y_{\theta} = Y\sin\theta + \tilde{Y}\cos\theta$, then $Y_{\theta}' = \frac{\td Y_{\theta}}{\td \theta} = Y\cos\theta - \tilde{Y} \sin\theta$. It is easy to verify that $Y_{\theta}$ and $Y_{\theta}'$ are independent. Note that $Y_0=\tilde{Y}$ and $Y_{\pi/2} = Y$. Then,
    \begin{equation*}
        F_u(X, Y) - F_u(X, \tilde{Y}) = \int_0^{\pi/2} \left\langle \nabla_{Y_{\theta}}F_u(X, Y_{\theta}), Y_{\theta}' \right\rangle \td\theta.
    \end{equation*}
    Equivalently, we have
    \begin{equation*}
        \E\left[ \left. \exp \left\{ t\left( F_u(X, Y) - F_u(X, \tilde{Y}) \right) \right\} \mathbb{I}_{B_Y} \right| X \right] = \E\left[ \left. \exp \left\{ t \int_0^{\pi/2} \left\langle \nabla_{Y_{\theta}}F_u(X, Y_{\theta}), Y_{\theta}' \right\rangle \td\theta \right\} \mathbb{I}_{B_Y} \right| X \right].
    \end{equation*}
    Applying Inequality \ref{eq:jensenBB},
    \begin{align*}
        \E\left[ \left. \exp\left\{ tF_u(X, Y) \right\} \mathbb{I}_{B_Y} \right| X \right] 
        \leq \E\left[ \left. \exp\left\{ t \int_0^{\pi/2} \left\langle \nabla_{Y_{\theta}}F_u(X, Y_{\theta}), Y_{\theta}' \right\rangle \td\theta \right\} \mathbb{I}_{B_Y} \right| X \right] .
    \end{align*}
    By Jensen's inequality and Fubini's theorem,
    \begin{align*}
        \E\left[ \left. \exp\left\{ tF_u(X, Y) \right\} \mathbb{I}_{B_Y} \right| X \right] 
        \leq \frac{2}{\pi} \int_0^{\pi/2} \E\left[ \left. \exp\left\{ \frac{t\pi}{2}  \left\langle \nabla_{Y_{\theta}}F_u(X, Y_{\theta}), Y_{\theta}' \right\rangle \right\} \mathbb{I}_{B_Y} \right| X \right] \td\theta .
    \end{align*}
    Since $Y_{\theta}'$ is a matrix with each entry following a standard Gaussian distribution independently, we can use the moment generating function of $Y_{\theta}'$ to derive
    \begin{align*}
        \E\left[ \left. \exp\left\{ tF_u(X, Y) \right\} \mathbb{I}_{B_Y} \right| X \right] 
        \leq \frac{2}{\pi} \int_0^{\pi/2} \E\left[ \left. \exp\left\{ \frac{t^2\pi^2}{8}  \left\| \nabla_{Y_{\theta}}F_u(X, Y_{\theta}) \right\|^2 \right\} \mathbb{I}_{B_Y} \right| X \right] \td\theta .
    \end{align*}
    As $A_X$ takes place and $\mathbb{I}_{B_Y}$ exists, by Inequality \ref{eq:normBB},
    \begin{align*}
        \E\left[ \left. \exp\left\{ tF_u(X, Y) \right\} \mathbb{I}_{B_Y} \right| X \right] 
        \leq \frac{2}{\pi} \int_0^{\pi/2} \E\left[ \left. \exp\left\{ t^2 Cnd\log{n}\log{d} \right\} \mathbb{I}_{B_Y} \right| X \right] \td\theta 
        \leq \exp\left\{ t^2 C_\gamma n^2\log^2{n} \right\} .
    \end{align*}
    By Markov's inequality,
    \begin{align*}
        \bbP\left( \left. F_u(X, Y) \geq C n^{\frac{3}{2}}\log^2{n} \right| X \right)
        \leq & \frac{\E\left[ \left. \exp\left\{ tF_u(X, Y) \right\} \mathbb{I}_{B_Y} \right| X \right]}{\exp\left\{ tCn^{\frac{3}{2}}\log{n} \right\}} + \bbP\left( B_Y^c \right) \\
        \leq & \exp\left\{ t^2 C_\gamma n^2\log^2{n} - tC n^{\frac{3}{2}}\log^2{n} \right\} + 2\exp\left\{ -C_0 n\log{n} \right\}
    \end{align*}
    Note that $d/n=\gamma \in (0, \infty)$. Taking $t = \frac{1}{\sqrt{n}\log{n}}$, there exist $C_\gamma>0$ such that
    \begin{equation*}
        \bbP\left( \left. F_u(X, Y) \geq C n^{\frac{3}{2}}\log^2{n} \right| X \right)
        \leq 3\exp\{ -C_\gamma n\log{n} \}.
    \end{equation*}
\end{proof}

\begin{prop} 
    \label{prop:BB}
    Suppose $X=(X_1,...,X_n)\in\reals^{d\times n}$ and $Y=(Y_1,...,Y_n)\in\reals^{d\times n}$ are independent matrices. The collection $\{ X_1,\ldots, X_n, Y_1,\ldots,Y_n \}$ consists of independent random vectors, each distributed as $\mathcal{N}(\zero, I_d)$. Define matrix $K(X, Y) \in \reals^{n\times n}$ as
    \begin{align*}
        K(X, Y)_{ij} = \left\{ \begin{array}{rcl}
            X_i^\top X_j Y_i^\top Y_j, & i\neq j; \\
            0, & i = j.
        \end{array} \right.
    \end{align*}
    Let $n, d \rightarrow \infty$ with $d/n = \gamma \in (0, \infty)$. Th en, with probability $1 - C_\gamma n^{-2}$,
    \begin{equation*}
        \|K(X, Y)\| \leq 2C n^{\frac{3}{2}}\log^2{n}.
    \end{equation*}
\end{prop}

\begin{proof}
    Let $N_{\epsilon}$ be the $\epsilon$-net of the unit sphere $S^{n-1}$. By \cite{Vershynin_2018}, the cardinality of $N_{\epsilon}$ is bounded by
    \begin{equation*}
        |N_{\epsilon}| \leq \left( 1 + \frac{2}{\epsilon} \right)^n;
    \end{equation*}
    also,
    \begin{equation*}
        \|K(X, Y)\| \leq \frac{1}{1-2\epsilon} \cdot \max_{u\in N_{\epsilon}} \left\langle K(X, Y)u, u \right\rangle.
    \end{equation*}
    Define the event $A_X$ as
    \begin{equation*}
        \max_{i\neq j} \left| X_i^\top X_j \right| \leq C \sqrt{d\log{d}}, \quad \|X\| \leq C\sqrt{n\log{n}} .
    \end{equation*}
    Take $\epsilon = \frac{1}{4}$, we have
    \begin{align*}
        \bbP\left( \|K(X, Y)\| \geq 2C n^{\frac{3}{2}}\log^2{n} \right)
        \leq & \E\left[ \bbP\left( \left. \max_{u\in N_{\epsilon}} \left\langle K(X, Y)u, u \right\rangle \geq C n^{\frac{3}{2}}\log^2{n} \right| X \right) \mathbb{I}_{A_X} \right] + \bbP\left( A_X^c \right) \\
        \leq & 9^n \cdot \E\left[ \bbP\left( \left. \left\langle K(X, Y)u, u \right\rangle \geq C n^{\frac{3}{2}}\log^2{n} \right| X \right) \mathbb{I}_{A_X} \right] + \bbP\left( A_X^c \right)
    \end{align*}
    Since entries of $X$ are independent standard Gaussian random vectors in $\reals^d$, \cite{Vershynin_2018} tells that the probability of $\|X\| \leq C\sqrt{n\log{n}}$ taking place is at least $1 - 2\exp\{ -C_0n\log{n} \}$. By Lemma \ref{lemma:AX}, we can directly deduce that $\bbP\left( A_X \right) \geq 1 - C_\gamma d^{-2}$. The inner product $\left\langle K(X, Y)u, u \right\rangle$ is the same as the $F_u(X, Y)$ function defined in Lemma \ref{lemma:BB}. With the existence of $\mathbb{I}_{A_X}$, conditions of Lemma \ref{lemma:BB} are satisfied. Then,
    \begin{align*}
        \bbP\left( \|K(X, Y)\| \geq 2C n^{\frac{3}{2}}\log^2{n} \right) 
        \leq 9^{n} \cdot 3\exp\left\{ -C_\gamma n\log{n} \right\} + C_\gamma d^{-2} 
        \leq \exp\{-C_\gamma n\log n\} + C_\gamma n^{-2} .
    \end{align*}
    With $n$ large enough, we have 
    \begin{equation*}
        \bbP\left( \|K(X, Y)\| \leq 2C n^{\frac{3}{2}}\log^2{n} \right) \geq 1 - C_\gamma n^{-2} .
    \end{equation*}
\end{proof}

\begin{lem}
    \label{lemma:GB}
    Suppose $X=(X_1,...,X_{n+m})\in\reals^{d\times (n+m)}$ and $Y=(Y_1,...,Y_n)\in\reals^{d\times n}$ are independent matrices. The collection $\{ X_1,\ldots, X_{n+m}, Y_1,\ldots,Y_n \}$ consists of independent random vectors, each distributed as $\mathcal{N}(\zero, I_d)$. Let $u\in \reals^n$ and $v\in \reals^m$ satisfy $\|u\|=\|v\|=1$, and $R\in\reals^{d\times d}$ be an orthogonal rotation matrix. Define the function $F_{uv}:\reals^{d\times n} \times \reals^{d\times n}\rightarrow \reals$,
    \begin{equation*}
        F_{uv}(X, Y) = \sum_{i\in [n]} \sum_{j\in[m]} u_i v_j \left( X_i^\top X_{n+j} \right) \left( X_i^\top R^\top Y_j \right),
    \end{equation*}
    Also, define the event $A_x$ as
    \begin{equation*}
        \max_{i\in[n],j\in[m]} \left| X_i^\top X_{n+j} \right| \leq C \sqrt{d\log{d}}, \quad \|X_{1:n}\|\leq C\sqrt{n\log{n}},
    \end{equation*}
    where $X_{1:n}$ formed by the first $n$ columns of $X$. Assume $A_X$ takes place. Let $n, m, d \rightarrow \infty$ with $d/n = \gamma_1 \in (0, \infty)$ and $d/m = \gamma_2\in(0, \infty)$. Then, for any fixed constant $C > 0$, there is an $C_\gamma > 0$ such that 
    \begin{equation*}
        \bbP\left( \left. F_{uv}(X, Y) \geq C n^{\frac{3}{2}}\log^2{n} \right| X \right) \leq \exp\{ -C_\gamma n\log{n} \}.
    \end{equation*}
\end{lem}

\begin{proof}
    Suppose $\tilde{Y}$ is an independent copy of $Y$. Define $\phi(y) = \E[\exp\{ t(F_{uv}(X, Y) - y) \} | X]$. Let $y = F_{uv}(X, \tilde{Y})$, then $\E[y|X]=0$. By Jensen's inequality,
    \begin{align} \label{eq:jensenGB}
        \E\left[ \left. \exp\left\{ tF_{uv}(X, Y) \right\} \right| X \right] 
        \leq \E\left[ \left. \exp\left\{ t \left( F_{uv}(X, Y) - F_{uv}(X, \tilde{Y}) \right) \right\} \right| X \right].
    \end{align}
    Next, we check
    \begin{equation*}
        \nabla_{Y_j}F_{uv}(X, Y) = \sum_{i\in [n]} u_i v_j(X_i^\top X_{n+j}) R X_i = v_j RX_{1:n} W_j,
    \end{equation*}
    where $W_j\in\reals^n$ with entries: $W_{j, i} = u_i(X_i^\top X_{n+j})$, $i\in[n]$. When $A_X$ takes place,
    \begin{align} \label{eq:normGB}
        \left\| \nabla_{Y} F_{uv}(X, Y) \right\|^2 
        \leq \|X_{1:n}\|^2 \sum_{j\in[m]} v_j^2 \|W_j\|^2 
        = \|X_{1:n}\|^2 \sum_{j\in[m]} \sum_{i\in[n]} u_i^2 v_j^2 (X_i^\top X_{n+j})^2 
        \leq C_{\gamma} n^2\log^2{n}.
    \end{align}
    Note $C_{\gamma}$ is a constant depending on $\gamma_1$ and $\gamma_2$, and its value may change line to line. Denote $Y_{\theta} = Y\sin\theta + \tilde{Y}\cos\theta$ and $Y_{\theta}' = \frac{\td Y_{\theta}}{\td \theta} = Y\cos\theta - \tilde{Y} \sin\theta$. It is easy to verify that $Y_{\theta}$ and $Y_{\theta}'$ are independent. Note that $Y_0=\tilde{Y}$ and $Y_{\pi/2} = Y$. Then,
    \begin{equation*}
        F_{uv}(X, Y) - F_{uv}(X, \tilde{Y}) = \int_0^{\pi/2} \left\langle \nabla_{Y_{\theta}}F_{uv}(X, Y_{\theta}), Y_{\theta}' \right\rangle \td\theta.
    \end{equation*}
    Equivalently, we have
    \begin{equation*}
        \E \left[ \left. \exp\left\{ t \left( F_{uv}(X, Y) - F_{uv}(X, \tilde{Y}) \right) \right\} \right| X \right]
        = \E \left[ \left. \exp\left\{ t \int_0^{\pi/2} \left\langle \nabla_{Y_{\theta}}F_{uv}(X, Y_{\theta}), Y_{\theta}' \right\rangle \td\theta \right\} \right| X \right].
    \end{equation*}
    Applying the Inequality \ref{eq:jensenGB},
    \begin{align*}
        \E\left[ \left. \exp\left\{ tF_{uv}(X, Y) \right\} \right| X \right]
        \leq \E\left[ \left. \exp\left\{ t \int_0^{\pi/2} \left\langle \nabla_{Y_{\theta}}F_{uv}(X, Y_{\theta}), Y_{\theta}' \right\rangle \td\theta \right\} \right| X \right].
    \end{align*}
    By Jensen's inequality and Fubini's theorem,
    \begin{align*}
        \E\left[ \left. \exp\left\{ tF_{uv}(X, Y) \right\} \right| X \right] 
        \leq \frac{2}{\pi} \int_0^{\pi/2} \E\left[ \left. \exp\left\{ \frac{t\pi}{2}  \left\langle \nabla_{Y_{\theta}}F_{uv}(X, Y_{\theta}), Y_{\theta}' \right\rangle \right\} \right| X \right] \td\theta.
    \end{align*}
    Since $Y_{\theta}'$ is a matrix with each entry following a standard Gaussian distribution independently, we can use the moment generating function of $Y_{\theta}'$ to derive
    \begin{align*}
        \E\left[ \left. \exp\left\{ tF_{uv}(X, Y) \right\} \right| X \right] 
        \leq \frac{2}{\pi} \int_0^{\pi/2} \E\left[ \left. \exp\left\{ \frac{t^2\pi^2}{8}  \left\| \nabla_{Y_{\theta}}F_{uv}(X, Y_{\theta}) \right\|^2 \right\} \right| X \right] \td\theta.
    \end{align*}
    As $A_X$ takes place, by Inequality \ref{eq:normGB},
    \begin{align*}
        \E\left[ \left. \exp\left\{ tF_{uv}(X, Y) \right\} \right| X \right] 
        \leq \frac{2}{\pi} \int_0^{\pi/2} \E\left[ \left. \exp\left\{ t^2 C_{\gamma} n^2\log^2{n} \right\} \right| X \right] \td\theta
        \leq \exp\left\{ t^2 C_{\gamma} n^2\log^2{n} \right\}.
    \end{align*}
    By Markov's inequality,
    \begin{align*}
        \bbP\left( \left. F_{uv}(X, Y) \geq C n^{\frac{3}{2}}\log^2{n} \right| X \right) 
        \leq & \frac{\E\left[ \left. \exp\left\{ tF_{uv}(X, Y) \right\} \right| X \right]}{\exp\left\{ tC n^{\frac{3}{2}}\log{n} \right\}} \\
        \leq & \exp\left\{ t^2 C_{\gamma} n^2\log^2{n} - tC n^{\frac{3}{2}}\log^2{n} \right\}
    \end{align*}
    Taking $t = \frac{1}{\sqrt{n}\log{n}}$, there exist $C_\gamma >0$ such that
    \begin{equation*}
        \bbP\left( \left. F_{uv}(X, Y) \geq C n^{\frac{3}{2}}\log^2{n} \right| X \right)
        \leq \exp\{ -C_\gamma n\log{n} \}.
    \end{equation*}
\end{proof}

\begin{prop}
    \label{prop:GB}
    Suppose $X=(X_1,...,X_{n+m})\in\reals^{d\times (n+m)}$ and $Y=(Y_1,...,Y_n)\in\reals^{d\times n}$ are independent matrices. The collection $\{ X_1,\ldots, X_{n+m}, Y_1,\ldots,Y_n \}$ consists of independent random vectors, each distributed as $\mathcal{N}(\zero, I_d)$. Define matrix $L(X, Y) \in \reals^{n\times m}$, with entries
    \begin{align*}
        L(X, Y)_{ij} = \left( X_i^\top X_{n+j} \right) \left( X_i^\top R^\top Y_j \right)
    \end{align*}
    Let $n, m, d \rightarrow \infty$ with $d/n = \gamma_1 \in (0, \infty)$ and $d/m = \gamma_2\in(0, \infty)$. Then, with probability at least $1-C_\gamma d^{-2}$,
    \begin{equation*}
        \|L(X, Y)\| \leq 2C n^{\frac{3}{2}}\log^2{n}.
    \end{equation*}
\end{prop}

\begin{proof}
    Let $N_{\epsilon}$ be the $\epsilon$-net of the unit sphere $S^{n-1}$ and $M_{\epsilon}$ be the $\epsilon$-net of the unit sphere $S^{m-1}$. Suppose $u^*\in S^{n-1}$ and $v^*\in S^{m-1}$ satisfies 
    \begin{equation*}
        \|L(X, Y)\| = \sup_{u\in S^{n-1}, v\in S^{m-1}} \langle L(X, Y) u, v \rangle = \langle L(X, Y) u^*, v^* \rangle.
    \end{equation*}
    We can always find $u\in N_{\epsilon}$ with $\|u - u^*\| \leq \epsilon$ and $v\in M_{\epsilon}$ with $\| v - v^* \| \leq \epsilon$. Hence, 
    \begin{align*}
        \langle L(X, Y) u^*, v^* \rangle 
        = & \langle L(X, Y) u, v \rangle + \langle L(X, Y) u, (v^* - v) \rangle + \langle L(X, Y) (u^* - u), v \rangle \\
        \leq & \langle L(X, Y) u, v \rangle + \epsilon \|L(X, Y)\| + \epsilon \|L(X, Y)\|
    \end{align*}
    Taking the maximum on both sides of the inequality, we get
    \begin{equation*}
        \|L(X, Y)\| \leq \frac{1}{1-2\epsilon} \cdot \max_{u\in N_{\epsilon}, v\in M_{\epsilon}} \left\langle L(X, Y)u, v \right\rangle.
    \end{equation*}
    By \cite{Vershynin_2018}, the cardinality of $N_{\epsilon}$ and $M_{\epsilon}$ are bounded,
    \begin{equation*}
        |N_{\epsilon}| \leq \left( 1 + \frac{2}{\epsilon} \right)^n; \quad |M_{\epsilon}| \leq \left( 1 + \frac{2}{\epsilon} \right)^m.
    \end{equation*}
    Define the event $A_x$ as
    \begin{equation*}
        \max_{i\in[n],j\in[m]} \left| X_i^\top X_{n+j} \right| \leq C \sqrt{d\log{d}}, \quad \|X_{1:n}\|\leq C\sqrt{n\log{n}},
    \end{equation*}
    where $X_{1:n}$ formed by the first $n$ columns of $X$. Since entries of $X_{1:n}$ are independent standard Gaussian random vectors in $\reals^d$, \cite{Vershynin_2018} tells that the probability of $\|X_{1:n}\| \leq C\sqrt{n\log{n}}$ taking place is at least $1 - 2\exp\{ -C_0n\log{n} \}$. By Lemma \ref{lemma:AX}, we can directly deduce that $\bbP\left( A_X \right) \geq 1 - C_\gamma d^{-2}$. Taking $\epsilon = \frac{1}{4}$, we have
    \begin{align*}
        \bbP\left( \|L(X, Y)\| \geq 2C_2n^{\frac{3}{2}}\log^2{n} \right) 
        \leq & \E\left[ \bbP\left( \left. \max_{u\in N_{\epsilon}, v\in M_{\epsilon}} \left\langle L(X, Y)u, v \right\rangle \geq C_2n^{\frac{3}{2}}\log^2{n} \right| X \right) \mathbb{I}_{A_X} \right] + \bbP\left( A_X^c \right) \\
        \leq & 9^{n + m} \cdot \E\left[ \bbP\left( \left. \left\langle L(X, Y)u, v \right\rangle \geq C_2n^{\frac{3}{2}}\log^2{n} \right| X \right) \mathbb{I}_{A_X} \right] + \bbP\left( A_X^c \right)
    \end{align*}
    The inner product $\left\langle L(X, Y)u, v \right\rangle$ is the same as the $F_{u,v}(X, Y)$ function defined in Lemma \ref{lemma:GB}. With the existence of $\mathbb{I}_{A_X}$, conditions of Lemma \ref{lemma:GB} are satisfied. Also, $d/n = \gamma_1 \in (0, \infty)$ and $d/m = \gamma_2\in(0, \infty)$ tells that $n+m=cn$ for some positive constant $c$. Then, by Lemma \ref{lemma:GB},
    \begin{align*}
        \bbP\left( \|L(X, Y)\| \geq 2C n^{\frac{3}{2}}\log^2{n} \right) 
        \leq 9^{cn} \cdot \exp\left\{ -C_\gamma n\log{n} \right\} + C_\gamma d^{-2}
    \end{align*}
    With $n$ large enough, there exists a constant $C_\gamma >0$ such that
    \begin{equation*}
        \bbP\left( \|L(X, Y)\| \leq 2Cn^{\frac{3}{2}}\log^2{n} \right) \geq 1 - C_\gamma d^{-2}
    \end{equation*}    
\end{proof}

\begin{lem} 
    \label{lemma:LM_chi}
    Let $X, Y \sim \mathcal{N}(\zero, I_d)$ be independent Gaussian random vectors in $\mathbb{R}^d$. Suppose $d > 8\log{n}$. Then,
    \begin{equation*}
    \label{eq:x4LR}
        \bbP\left( \left| \|X\|^4 - d^2 - 2d \right| \geq Cd^{\frac{3}{2}}\sqrt{\log{n}} \right) 
        \leq 2n^{-2}.
    \end{equation*}
    Moreover, 
    \begin{equation*}
        \bbP\left( \left| \|X\|^2\|Y\|^2 - d^2 \right| \geq Cd^{\frac{3}{2}}\sqrt{\log{n}} \right) 
        \leq 4n^{-2} .
    \end{equation*}
\end{lem}

\begin{proof}
    Let $W\sim \mathcal{X}_d^2$. \cite{Laurent} states that for any $t>0$, the tail probability of $W$ satisfies 
    \begin{equation*}
        \bbP\left( W \geq d + 2\sqrt{td} + 2t \right) \leq e^{-t}
        \quad \text{and} \quad
        \bbP\left( W \leq d - 2\sqrt{td} \right) \leq e^{-t}.
    \end{equation*}
    Note that $\|X\|^2$ and $\|Y\|^2$ are chi-square random variables with degree of freedom $d$. By $t=2\log{n}$, we have
    \begin{equation*}
        \bbP\left( \|X\|^4 \geq \left( d + 2\sqrt{2d\log{n}} + 4\log{n} \right)^2 \right) \leq n^{-2},
    \end{equation*}
    and
    \begin{equation*}
        \bbP\left( \|X\|^4 \leq \left( d - 2\sqrt{2d\log{n}} \right)^2 \right) \leq n^{-2}.
    \end{equation*}
    There exists a positive constant $C$ such that
    \begin{align*}
        \left( d + 2\sqrt{2d\log{n}} + 4\log{n} \right)^2 - d^2 - 2d 
        \leq Cd^{\frac{3}{2}}\sqrt{\log{n}}
    \end{align*}
    With $d>8\log{n}$, we have $d - 2\sqrt{2\log{n} d}>0$. Then,
    \begin{align*}
        & \bbP\left( \left| \|X\|^4 - d^2 - 2d \right|
        \geq C d^{\frac{3}{2}}\sqrt{\log{n}} \right) \\
        \leq & \bbP\left( \left| \|X\|^4 - d^2 - 2d \right|
        \geq \left( d + 2\sqrt{2d\log{n}} + 4\log{n} \right)^2 - d^2 - 2d \right) \\
        \leq & \bbP\left( \|X\|^4 \geq \left( d + 2\sqrt{2d\log{n}} + 4\log{n} \right)^2 \right) + \bbP\left( \|X\|^4 \leq \left( d - 2\sqrt{2d\log{n}} \right)^2 \right) \\
        \leq & 2n^{-2}
    \end{align*}
    
    Similarly, when $t=2\log{n}$, 
    \begin{equation*}
        \bbP\left( \|X\|^2\|Y\|^2 \geq \left(d + 2\sqrt{2d\log{n}} + 4\log{n}\right)^2 \right) \leq 2n^{-2},
    \end{equation*}
    and 
    \begin{equation*}
        \bbP\left( \|X\|^2\|Y\|^2 \leq \left( d - 2\sqrt{2d\log{n}} \right)^2  \right) \leq 2n^{-2}
    \end{equation*}
    There exists a positive constant $C$ such that
    \begin{align*}
        \left( d + 2\sqrt{2d\log{n}} + 4\log{n} \right)^2 - d^2 
        \leq C d^{\frac{3}{2}}\sqrt{\log{n}}
    \end{align*}
    With $d>8\log{n}$, we get the second probability bound
    \begin{align*}
        & \bbP\left( \left| \|X\|^2\|Y\|^2 - d^2 \right|
        \geq C d^{\frac{3}{2}}\sqrt{\log{n}} \right) \\
        \leq & \bbP\left( \left| \|X\|^2\|Y\|^2 - d^2 \right|
        \geq \left( d + 2\sqrt{2d\log{n}} + 4\log{n} \right)^2 - d^2\right) \\
        \leq & \bbP\left( \|X\|^2\|Y\|^2 \geq \left( d + 2\sqrt{2d\log{n}} + 4\log{n} \right)^2 \right) + \bbP\left( \|X\|^2\|Y\|^2 \leq \left( d - 2\sqrt{2d\log{n}} \right)^2 \right) \\
        \leq & 4n^{-2}
    \end{align*}
\end{proof}

\section{Proof of Theorem \ref{theo:mainEig}}\label{appendix:theo_eig}

\begin{proof}[Proof of Theorem \ref{theo:mainEig}]
    WLOG, we rearrange the columns of $X$ so that the first $|G|$ columns are $G$ indexed and the last $|B|$ columns are $B$ indexed. The matrix $Y$ rearranges its columns according to $X$.

    Let $E = H - \E[H]$. Denote the upper left $|G|\times |G|$ submatrix of $E$ as $E_G$, the lower right $|B|\times |B|$ submatrix of $E$ as $E_B$, and the off-diagonal submatrices with size $|G|\times |B|$ and $|B|\times |G|$ as $E_{GB}$ and $E_{BG}$, respectively. Then,
    \begin{equation*}
        \|E\| \leq \|E_G\| + \|E_B\| + 2\|E_{GB}\|.
    \end{equation*}
    
    Firstly, we find the upper bound of $\|E_G\|$. Rewrite $E_G = D_G + \tilde{E}_G$, where $D_G$ is the diagonal entries of $E_G$ and $\tilde{E}_G$ is the off-diagonal entries of $E_G$. Define a kernel matrix $K(X) \in \reals^{|G|\times|G|}$
    \begin{align*}
        K(X)_{ij} = \left\{ \begin{array}{rcl}
            \frac{1}{\sqrt{d}} k\left( \frac{1}{\sqrt{d}}X_i^\top X_j \right), & i\neq j \\
            0, & i=j 
        \end{array} \right. ,
    \end{align*}
    where $k(x) = x^2 - 1$. By \cite{Fan}, when $d/n=\gamma\in(0, \infty)$, the following holds with high probability $1-o(1)$,
    \begin{equation*}
        \|K(X)\| \leq C_\gamma.
    \end{equation*}
    Note that $\tilde{E}_G = d^{\frac{3}{2}} K(X)$, then
    \begin{equation*}
        \|\tilde{E}_G\| \leq C_\gamma d^{\frac{3}{2}}.
    \end{equation*}
    The condition $d>8\log{n}$ in Lemma \ref{lemma:LM_chi} holds automatically when $d, n\rightarrow \infty$, then with probability $1 - 2n^{-2}$,
    \begin{equation*}
        \|D_G\| = \max_{i\in G} \left| \|X_i\|^4 - d^2 - 2d \right| \leq Cd^{\frac{3}{2}}\sqrt{\log{n}}.
    \end{equation*}
    Overall, with high probability $1-o(1)$,
    \begin{equation*}
        \|E_G\| \leq \|D_G\| + \|\tilde{E}_G\| \leq C_\gamma d^{\frac{3}{2}}\sqrt{\log{n}}
    \end{equation*}

    Next, we consider $E_B$. Similarly, we decompose $E_B = D_B + \tilde{E}_B$, where $D_B$ is the diagonal entries of $E_B$ and $\tilde{E}_B$ is the off-diagonal entries of $E_B$. The condition $d>8\log{n}$ in Lemma \ref{lemma:LM_chi} holds automatically when $d, n\rightarrow \infty$, we can use Lemma \ref{lemma:LM_chi} to upper bound $\|D_B\|$. With probability $1 - 4n^{-2}$,
    \begin{equation*}
        \|D_B\| = \max_{i\in B} \left| \|X_i\|^2 \|Y_i\|^2 - d^2 \right| \leq Cd^{\frac{3}{2}}\sqrt{\log{n}}
    \end{equation*}
    Followed by Proposition \ref{prop:BB}, we have
    \begin{equation*}
        \bbP\left( \|\tilde{E}_B\| \leq Cn^{\frac{3}{2}}\log^2{n} \right) \geq 1 - C_\gamma d^{-2}.
    \end{equation*}
    With $d/n=\gamma\in(0, \infty)$, we conclude
    \begin{equation*}
        \|E_B\| 
        \leq \|D_B\| + \|\tilde{E}_B\| 
        \leq Cn^{\frac{3}{2}}\log^2{n}
    \end{equation*}
    with probability $1-C_\gamma n^{-2}$.

    Lastly, we are left to consider the off-diagonal blocks. By Proposition \ref{prop:GB}, with probability $1 - C_\gamma d^{-2}$,
    \begin{equation*}
        \|E_{GB}\| \leq Cn^{\frac{3}{2}}\log^2{n}.
    \end{equation*}
    Since $E_{BG} = E_{GB}^\top$, we have
    \begin{equation*}
        \|E_{BG}\| = \|E_{GB}\| \leq Cn^{\frac{3}{2}}\log^2{n}.
    \end{equation*}
    
    Combining all results, with high probability $1-o(1)$,
    \begin{equation*}
        \|E\| \leq C_\gamma \left( d^{\frac{3}{2}}\sqrt{\log{n}} + n^{\frac{3}{2}}\log^2{n} \right) \leq C_\gamma n^{\frac{3}{2}}\log^2{n}.
    \end{equation*}
\end{proof}

\section{Proof of Theorem \ref{theo:mainRS}}\label{appendix:theo_rs}

\begin{proof}[Proof of Theorem \ref{theo:mainRS}]
    Suppose $z\sim\mathcal{N}(0, 1)$. Denote $K:=\|z\|_{\psi_2}^2$. The tail probability of the summation of independent centred sub-exponential random variables is given by \cite{Vershynin_2018}: denote $K_i = \|X_i\|_{\psi_1}$, for any $t>0$,
    \begin{align} \label{eq:versh}
        \bbP\left( \left| \sum_{i=1}^n X_i \right| \geq t \right)
        \leq & 2\exp\left[ - c \min \left( \frac{t^2}{\sum_{i=1}^n K_i^2}, \frac{t}{\max_{i\in[n]} K_i } \right) \right]
    \end{align}

    Case $i\in G$: The deviation of row sum $S_i$ from its mean for $i\in G$ can be written as
    \begin{align*}
        \left| S_i - \E[S_i] \right| 
        & = \left| \left( \|X_i\|^4 - d^2 - 2d \right) + \sum_{j\in G\backslash \{i\}} \left( \left( X_i^\top X_j \right)^2 - d \right) + \sum_{j\in B} X_i^\top X_j X_i^\top R' Y_j \right| \\
        & = \left| \left( \|X_i\|^4 - d^2 - 2d \right) + \sum_{j\in G\backslash \{i\}} \left( \left( X_i^\top X_j \right)^2 - \|X_i\|^2 + \|X_i\|^2 - d \right) + \sum_{j\in B} X_i^\top X_j X_i^\top R' Y_j \right| \\
        & \leq \left| \|X_i\|^4 - d^2 - 2d \right| + r n \left| \|X_i\|^2 - d \right| \\
        & \hspace{2cm} + \|X_i\|^2 \left| \sum_{j\in G\backslash \{i\}} \left( \left( \frac{X_i^\top X_j}{\|X_i\|} \right)^2 - 1 \right) \right| + \|X_i\|^2 \left| \sum_{j\in B} \frac{X_i^\top X_j}{\|X_i\|} \frac{X_i^\top R' Y_j}{\|X_i\|} \right|
    \end{align*}
    With $\log n \lesssim d$, Lemma \ref{lemma:LM_chi} gives,
    \begin{align*}
        \bbP\left( \left| \|X\|^4 - d^2 - 2d \right| \geq Cd^{\frac{3}{2}}\sqrt{\log{n}} \right) 
        \leq 2n^{-2}.
    \end{align*}
    Also, by chi-square tail probability \cite{Laurent},
    \begin{align*}
        \bbP\left( \left| \|X_i\|^2 - d \right| \geq C \sqrt{d\log{n}} \right) \leq 2n^{-2}
    \end{align*}
    and
    \begin{align*}
        \bbP\left( \|X_i\|^2 \geq Cd \right) \leq n^{-2}.
    \end{align*}
    For $j\in G \backslash \{i\}$, we have $\left( \frac{X_i^\top X_j}{\|X_i\|} \right)^2 \sim \mathcal{X}_1^2$ and $\left\| \left( \frac{X_i^\top X_j}{\|X_i\|} \right)^2 - 1 \right\|_{\psi_1}\leq CK$. Given $X_i$, the random variables $\left( \frac{X_i^\top X_j}{\|X_i\|} \right)^2 - 1$ are iid centered sub-exponential random variables. By Inequality \ref{eq:versh}, the following probability tail follows when choosing $t = C K\sqrt{\frac{2 n\log{n}}{c_1}}$,
    \begin{align*}
        \bbP\left( \left| \sum_{j\in G\backslash \{i\}} \left( \left( \frac{X_i^\top X_j}{\|X_i\|} \right)^2 - 1 \right) \right| \geq C K\sqrt{\frac{2 n\log{n}}{c_1}} \right)
        \leq 2 \exp\left[ -c_1 \min \left( \frac{2\log{n}}{c_1}, \sqrt{\frac{2 n \log{n}}{c_1}} \right) \right].
    \end{align*}
    For $j\in B$ and given $X_i$, $\frac{X_i^\top X_j}{\|X_i\|} \frac{X_i^\top R^\top Y_j}{\|X_i\|}$ are independent sub-exponential random variables with sub-exponential norm $\left\| \frac{X_i^\top X_j}{\|X_i\|} \frac{X_i^\top R^\top Y_j}{\|X_i\|} \right\|_{\psi_1}\leq CK$. Applying the sub-exponential tail probability bound in Inequality \ref{eq:versh} with $t = C K\sqrt{\frac{2 n\log{n}}{c_2}}$, we have
    \begin{align*}
        \bbP\left( \left| \sum_{j\in B} \frac{X_i^\top X_j}{\|X_i\|} \frac{X_i^\top R^\top Y_j}{\|X_j\|} \right| \geq CK\sqrt{\frac{2 n\log{n}}{c_2}} \right)
        \leq 2 \exp\left[ -c_2 \min \left( \frac{2\log{n}}{c_2}, \sqrt{\frac{2 n \log{n}}{c_2}} \right) \right].
    \end{align*}
    Let $c = \min\{c_1, c_2\}$. Exist $n_0$ such that for all $n>n_0$,
    \begin{equation*}
        \bbP\left( \left| \sum_{j\in G\backslash \{i\}} \left( \left( \frac{X_i^\top X_j}{\|X_i\|} \right)^2 - 1 \right) \right| \geq C K\sqrt{n\log{n}} \right) \leq 2 n^{-2}, 
    \end{equation*}
    and
    \begin{equation*}
        \bbP\left( \left| \sum_{j\in B} \frac{X_i^\top X_j}{\|X_i\|} \frac{X_i^\top R^\top Y_j}{\|X_j\|} \right| \geq CK\sqrt{n\log{n}} \right) \leq 2 n^{-2}.
    \end{equation*}
    Combining the above results: with probability at least $1 - 9n^{-2}$,
    \begin{align*}
        \left| S_i - \E[S_i] \right| \leq C_1 \sqrt{d\log{n}} \left( d + n + \sqrt{nd} \right).
    \end{align*}

    Case $i\in B$: The deviation of row sum $S_i$ from its mean for $i\in B$ can be written as
    \begin{align*}
        \left| S_i - \E[S_i] \right| 
        & = \left| \left( \|X_i\|^2 \|Y_i\|^2 - d^2 \right) + \sum_{j\in G} \left( X_i^\top X_j Y_i^\top RX_j \right) + \sum_{j\in B\backslash \{i\}} X_i^\top X_j Y_i^\top Y_j \right| \\
        & \leq \left| \|X_i\|^2 \|Y_i\|^2 - d^2 \right| + \|X_i\| \|Y_i\| \left| \sum_{j\in G} \frac{X_i^\top X_j}{\|X_i\|} \frac{Y_i^\top R X_j}{\|Y_i\|} \right| + \|X_i\| \|Y_i\| \left| \sum_{j\in B\backslash \{i\}} \frac{X_i^\top X_j}{\|X_i\|} \frac{Y_i^\top Y_j}{\|Y_i\|} \right|
    \end{align*}
    Since the condition $d>8\log{n}$ is satisfied, we can use Lemma \ref{lemma:LM_chi},
    \begin{align*}
        \bbP\left( \left| \|X_i\|^2 \|Y_i\|^2 - d^2 \right| \geq C d^{\frac{3}{2}}\sqrt{\log{n}} \right) 
        \leq 4n^{-2}.
    \end{align*}
    Also, the intermediate step of the proof of Lemma \ref{lemma:LM_chi} gives
    \begin{equation*}
        \bbP\left( \|X_i\| \|Y_i\| \geq d + 2\sqrt{2d\log{n}} + 4\log{n} \right) \leq 2n^{-2}
    \end{equation*}
    Then, there exists some positive constant $C$ such that
    \begin{align*}
        \bbP\left( \|X_i\| \|Y_i\| \geq Cd \right) \leq 2n^{-2}.
    \end{align*}
    Given $X_i$ and $Y_i$, the sub-exponential random variables $\frac{X_i^\top X_j}{\|X_i\|} \frac{Y_i^\top R X_j}{\|Y_i\|}$, $j\in G$, are independent and identically distributed with mean zero and sub-exponential norm $\left\| \frac{X_i^\top X_j}{\|X_i\|} \frac{Y_i^\top R X_j}{\|Y_i\|} \right\|_{\psi_1}\leq CK$. Now, apply Inequality \ref{eq:versh} with $t = C K\sqrt{\frac{2 n\log{n}}{c_1}}$ to get
    \begin{align*}
        \bbP\left( \left| \sum_{j\in G} \frac{X_i^\top X_j}{\|X_i\|} \frac{Y_i^\top R X_j}{\|Y_i\|} \right| \geq C K\sqrt{\frac{2 n\log{n}}{c_1}} \right)
        \leq 2 \exp\left[ -c_1 \min \left( \frac{2\log{n}}{c_1}, \sqrt{\frac{2 n \log{n}}{c_1}} \right) \right].
    \end{align*}
    for $j\in B\backslash\{i\}$ and given $X_i$ and $Y_i$, the sub-exponential random variables $\frac{X_i^\top X_j}{\|X_i\|} \frac{Y_i^\top Y_j}{\|Y_i\|}$ are independent and identically distributed with mean zero and sub-exponential norm $\left\| \frac{X_i^\top X_j}{\|X_i\|} \frac{Y_i^\top Y_j}{\|Y_i\|} \right\|_{\psi_1}\leq CK$. Apply Inquality \ref{eq:versh} with $t = CK\sqrt{\frac{2 n\log{n}}{c_2}}$, we have
    \begin{align*}
        \bbP\left( \left| \sum_{j\in B\backslash \{i\}} \frac{X_i^\top X_j}{\|X_i\|} \frac{Y_i^\top Y_j}{\|Y_i\|} \right| \geq CK\sqrt{\frac{2 n\log{n}}{c_2}} \right)
        \leq 2 \exp\left[ -c_2 \min \left( \frac{2\log{n}}{c_2}, \sqrt{\frac{2 n \log{n}}{c_2}} \right) \right].
    \end{align*}
    Let $c = \min\{c_1, c_2\}$. Exist $n_0$ positive such that for all $n > n_0$, 
    \begin{equation*}
        \bbP\left( \left| \sum_{j\in G} \frac{X_i^\top X_j}{\|X_i\|} \frac{Y_i^\top R X_j}{\|Y_i\|} \right| \geq C K\sqrt{n\log{n}} \right) \leq 2 n^{-2}
    \end{equation*}
    and
    \begin{equation*}
        \bbP\left( \left| \sum_{j\in B\backslash \{i\}} \frac{X_i^\top X_j}{\|X_i\|} \frac{Y_i^\top Y_j}{\|Y_i\|} \right| \geq CK\sqrt{n\log{n}} \right) \leq 2 n^{-2}.
    \end{equation*}
    Combining the above results: with probability at least $1 - 10n^{-2}$,
    \begin{align*}
        \left| S_i - \E[S_i] \right| \leq C_2 d\sqrt{\log{n}} \left( \sqrt{d} + \sqrt{n} \right).
    \end{align*}
\end{proof}

\end{document}